%% file: main.tex
\documentclass[submission]{eptcs}
\providecommand{\event}{QPL 2018} 
\usepackage{underscore}           
\usepackage{comment}

\usepackage{color}

\usepackage{array}
\usepackage{tabularx}
\usepackage{amsmath,amsthm,amssymb}

\usepackage[title]{appendix}

\usepackage{tikz}
\makeatletter
\include{tikzstyles}

\usetikzlibrary{decorations.markings}
\usetikzlibrary{shapes.geometric}
\pgfdeclarelayer{edgelayer}
\pgfdeclarelayer{nodelayer}
\pgfsetlayers{edgelayer,nodelayer,main}

\tikzstyle{none}=[inner sep=0pt]
\definecolor{hexcolor0xff0000}{rgb}{1.000,0.000,0.000}
\definecolor{hexcolor0x000000}{rgb}{0.000,0.000,0.000}
\definecolor{hexcolor0x00ff00}{rgb}{0.000,1.000,0.000}
\definecolor{hexcolor0x000000}{rgb}{0.000,0.000,0.000}
\definecolor{hexcolor0xffff00}{rgb}{1.000,1.000,0.000}
\definecolor{hexcolor0x000000}{rgb}{0.000,0.000,0.000}

\tikzstyle{rn}=[circle,fill=hexcolor0xff0000,draw=hexcolor0x000000,line width=0.8 pt]
\tikzstyle{gn}=[circle,fill=hexcolor0x00ff00,draw=hexcolor0x000000,line width=0.8 pt]
\tikzstyle{yn}=[circle,fill=hexcolor0xffff00,draw=hexcolor0x000000,line width=0.8 pt]

\tikzstyle{simple}=[-,draw=hexcolor0x000000,line width=2.000]
\tikzstyle{arrow}=[-,draw=hexcolor0x000000,postaction={decorate},decoration={markings,mark=at position .5 with {\arrow{>}}},line width=2.000]
\tikzstyle{tick}=[-,draw=hexcolor0x000000,postaction={decorate},decoration={markings,mark=at position .5 with {\draw (0,-0.1) -- (0,0.1);}},line width=2.000]

\theoremstyle{plain}
\newtheorem{theorem}{Theorem}

\newtheorem{definition}[theorem]{Definition}
\newtheorem{proposition}[theorem]{Proposition}
\newtheorem{lemma}[theorem]{Lemma}

\newcommand{\nn}{\nonumber\\}
\newcommand{\ket}[1]{| #1 \rangle}

\newcommand{\xp}{X'}

\newcommand{\zp}{Z'}

\title{Parallel Self-Testing of the GHZ State \\ with a Proof by Diagrams}

\author{Spencer Breiner
\institute{NIST}
\email{spencer.breiner@nist.gov}
\and
Amir Kalev
\institute{Joint Center for Quantum 
Information \\ and Computer Science (QuICS)}
\email{amirk@umd.edu}
\and
Carl A. Miller
\institute{NIST}
\institute{Joint Center for Quantum 
Information \\ and Computer Science (QuICS)}
\email{camiller@umd.edu}
}
\def\titlerunning{Parallel Self-Testing of the GHZ State with a Proof by Diagrams}
\def\authorrunning{S.~Breiner, A.~Kalev \& C.~A.~Miller}
\begin{document}
\maketitle

\begin{abstract}
Quantum self-testing addresses the following question: is it possible to verify the existence of a multipartite state even when one's measurement devices are completely untrusted? This problem has seen abundant activity in the last few years, particularly with the advent of parallel self-testing (i.e., testing several copies of a state at once), which has applications not only to quantum cryptography but also quantum computing. In this work we give the first error-tolerant parallel self-test in a three-party (rather than two-party) scenario, by showing that an arbitrary number of copies of the GHZ state can be self-tested. In order to handle the additional complexity of a three-party setting, we use a diagrammatic proof based on categorical quantum mechanics, rather than a typical symbolic proof. The diagrammatic approach allows for manipulations of the complicated tensor networks that arise in the proof, and gives a demonstration of the importance of picture-languages in quantum information. \end{abstract}

\section{Introduction}

Quantum rigidity has its origins in quantum key distribution, which is one of the original problems in quantum cryptography. 
In the 1980s Bennett and Brassard proposed a protocol for secret key distribution across an untrusted public quantum channel \cite{bennett2014quantum}.  A version of the Bennett-Brassard protocol can be expressed as follows: Alice prepares $N$ EPR pairs, and shares the second half of these pairs with Bob through the untrusted public channel.  Alice and Bob then perform random measurements on the resulting state, and check the results to verify that indeed their shared state approximates $N$ EPR pairs.  If these tests succeed, Alice and Bob then use other coordinated measurement results as the basis for their shared key.  Underlying the proof of security for the Bennett-Brassard protocol is the idea that if a shared $2$-qubit state approximates the behavior of a Bell state under certain measurements, then the state itself must approximate a Bell state. 

If we wish to deepen the security, we can ask: what if Alice's and Bob's measurement devices are also not trusted?  Can we prove security at a level that guards against possible exploitation of defects in their measurement devices?  This leads to the question of quantum rigidity: is it possible to completely verify the behavior of $n$ untrusted quantum measurement devices, based only on statistical observation of their measurement outputs, and without any prior knowledge of the state they contain?

We say that an $n$-player cooperative game is \textit{rigid} if an optimal score at that game guarantees that the players must have used a particular state and particular measurements.  We say that a state is \textit{self-testing} if its existence can be guaranteed by such a game.  Early results on this topic focused on self-testing the $2$-qubit Bell state \cite{popescu1992states, mayers1998quantum, mckague2012robust}.  Since then a plethora of results on other games and other states have appeared.  The majority of works have focused on the bipartite case, and there are a smaller number of works that address $n$-partite states for $n \geq 3$ \cite{miller2013optimal,mckague2016interactive, wu2014robust, yao2016self-testing, Pal:2014,
fadel2017}.

More recently, it has been observed that rigid games exist that self-test not only one copy of a bipartite state, but several copies at once.  Such games are a resource not only for cryptography, but also for quantum computation: these games can be manipulated to force untrusted devices to perform measurements on copies of the Bell state which carry out complex circuits.  This idea originated in \cite{reichardt2013classical} and has seen variants and improvements since then \cite{mckague2016interactive,natarajan2017quantum,coladangelo2017verifier}.  For such applications, it is important that the result include an error term which is (at most) bounded by some polynomial function of the number of copies of the state.\footnote{This condition
allows the computations to be performed in polynomial time.  The works \cite{natarajan2017quantum,coladangelo2017verifier} go
further, and prove an error
term that is independent
of the number of copies of the state.}


It is noteworthy that all results proved so far for error-tolerant self-testing of several copies of a state at once (that is, parallel self-testing) apply to bipartite states only \cite{mckague2016self,mckague2017self,
coudron2016parallel,natarajan2017quantum,natarajan2018low,ostrev2016structure,chao2016test,reichardt2013classical,coladangelo2017robust,coladangelo2017parallel,coladangelo2017all,coladangelo2018generalization}.  There is a general multipartite self-testing result in \cite{supic2017simple} which can be applied to the parallel case, but it is not error-tolerant and no explicit game is given.
Complexity of proofs is a factor in establishing new results in this direction: while it would be natural to extrapolate existing parallelization techniques to prove self-tests for $n$-partite states, the proofs for the bipartite case are already difficult, and we can expect that the same proofs for $n \geq 3$ are more so.  Yet, if this is an obstacle it is one worth overcoming, since multi-partite states are an important resource in cryptography.  For example, a much-cited paper in 1999 \cite{hillery1999quantum} proved that secret sharing is possible using several copies of the GHZ state
$\left| GHZ \right> = \frac{1}{\sqrt{2}} \left( \left| 000 \right> + \left| 111 \right> \right)$, in analogy to the use of Bell states in the original QKD protocol \cite{bennett2014quantum}.

In this work, we give the first proof of an error-tolerant parallel \textit{tripartite} self-test.  
Specifically, we prove that a certain class of $3$-player games self-tests $N$ copies of the GHZ state, for any $N \geq 1$, with an error term
that grows polynomially with $N$.  
To accomplish this we introduce, for the first time, the graphical language of categorical quantum mechanics into the topic of rigidity.  As we will discuss below, 
the use of graphical languages is a critical feature of the proof --- games involving more than $2$ players
involve complicated tensor networks, which are not easily expressed symbolically.
Our result thus demonstrates the power and importance of visual formal reasoning in quantum information processing.

\subsection{Categorical quantum mechanics}

Category theory is a branch of abstract mathematics which studies systems of interacting processes. In Categorical Quantum Mechanics (CQM), categories (specifically \emph{symmetric monoidal categories}) are used to represent and analyze the interaction of quantum states and processes.

Inspired by methods from computer science, CQM introduces an explicit distinction between traditional quantum semantics in Hilbert spaces and the syntax of quantum protocols and algorithms. In particular, symmetric monoidal categories support a diagrammatic formal syntax called \emph{string diagrams}, which provide an intuitive yet rigorous means for defining and analyzing quantum processes, in place of the more traditional bra-ket notation. This expressive notation helps to clarify definitions and proofs, making them easier to read and understand, and encourages the use of equational (rather than calculational) reasoning.

The origins of CQM's graphical methods can be found in Penrose's tensor diagrams \cite{penrose1971applications}, although earlier graphical languages from physics (Feynman diagrams) and computer science (process charts) can be interpreted in these terms. Later, Joyal and Street \cite{joyal1991geometry} used category theory and topology to formalize these intuitive structures. More recently, the works of Selinger \cite{selinger2004towards,selinger2007dagger} and Coecke, et al. \cite{abramsky2004categorical, coecke2006quantum} have substantially tightened the connection between categorical methods and quantum information, in particular developing diagrammatic approaches to  positive maps and quantum-classical interaction, respectively. A thorough and self-contained introduction to this line of research can be found in the recent textbook \cite{coecke2017picturing}.

For a brief review of categorical quantum mechanics and the syntax of string diagrams used in our proofs, see Appendix \ref{appendix cqm}.

\subsection{Statement of main result}

\label{statementsubsec}

In the three-player GHZ game, a referee chooses a random bit string
$xyz \in \{ 0, 1 \}^3$ such
that $x \oplus y \oplus z$=0, and distributes $x, y, z$ to the three 
players (Alice, Bob, Charlie) respectively, who return numbers $a, b, c \in \{ -1, 1 \}$ respectively.  The game is won if
\begin{eqnarray}
\label{ghzcondition}
abc & = & (-1)^{\neg (x \vee y \vee z)}.
\end{eqnarray}
(In other words, the game is
won if the parity of the outputs is even and $xyz \neq 000$, or the parity of the outputs is odd and $xyz = 000$.)  It is easy to prove that this game has no classical winning strategy.  On the other hand, if Alice, Bob, and Charlie share the $3$-qubit state
\begin{eqnarray}
\label{gstate}
\left| G \right> & = & 
\frac{1}{2 \sqrt{2}} \left( \sum_{r + s + t \leq 1} \left| rst \right> 
- \sum_{r + s + t \geq 2} \left| rst \right> \right),
\end{eqnarray}
and obtain their outputs by either performing an $X$-measurement $\left\{
\left|  + \right> \left< 
+ \right|, \left|  - \right> \left< 
- \right| \right\}$ on input
$0$ or the $Z$-measurement
$\left\{
\left|  0 \right>
\left<  0 \right|, \left| 
1 \right> 
\left< 
1 \right|
\right\}$ on input 1, they win perfectly.
This is known to be the only optimal
strategy (up to local changes of basis) and therefore the GHZ game is
rigid \cite{mckague2011self}. (An equivalent strategy, which is more conventional, is to use the state  $ \left| GHZ \right> =  \frac{1}{\sqrt{2}}
\left( \left| 000 \right > 
+ \left| 111 \right> \right)$
and $X$- and $Y$-measurements
to win the GHZ game.  We will use the 
state $\left| G \right>$ instead for
compatibility
with previous work on rigidity.
Note that $\left| G \right>$
is equivalent under local unitary transformations to $\left| GHZ \right>$.)

To extend this game to a self-test
for the $\left| G \right>^{\otimes N}$ state,
we use a game modeled after
\cite{chao2016test}.  The game requires the players to simulate playing the GHZ game $N$ times.
We give input to the $r$th player in the form of a pair $(U_r, f_r )$ where
$\{ 1, 2, \ldots, N \} \supseteq U_r \stackrel{f_r}{\longrightarrow} \{ 0, 1 \}$
is a partial function assigning an ``input'' value for some subset of the game's ``rounds'' $1,\ldots, N$. 
The output given by such a player is a function $g_r \colon U_r \to \{ 0, 1 \}$ assigning a bit-valued ``output'' to each round.  The game is won if the 
GHZ condition (\ref{ghzcondition}) is satisfied on all the rounds in $U_1 \cap U_2 \cap U_3$ for which the input
string was even-parity.

Fortunately, it is not necessary to
query the players on all possible subsets $U_r \subseteq \{ 0, 1, \ldots , N \}$ (which would involve an exponential number of inputs) --- it is only necessary to query them on one- and two-element subsets.  This
yields the game $\overline{GHZ}_N$, which is formally defined in
Figure~\ref{fig:centralgame}
below.  Our main result is the following.  (See Proposition~\ref{prop:state-approx multi} below for a formal statement.)

\vskip0.1in

\fbox{\parbox{6in}{

\begin{theorem}
The game $\overline{GHZ}_N$ is a self-test for the state $\left| G \right>^{\otimes N}$, with error
term ${\cal O} (  N^4 \sqrt{\epsilon} )$.
\end{theorem}
}}

\vskip0.1in

In other words, if three devices succeed at the game $\overline{GHZ}_N$ with probability $1 - \epsilon$,
then the devices must contain a state that approximates the state $\left| G \right>^{\otimes N}$ up to an error term of ${\cal O} (  N^4 \sqrt{\epsilon} )$.  
The proof proceeds by assuming that the players have such a high-performing strategy, 
and then using the measurements from that strategy to map their state isometrically to a state that is approximately of the form $\left| G \right> \otimes L'$,
where $L'$ is some arbitrary tripartite ``junk'' state.
This approach is a graphical translation
of the method of many previous works on rigidity (in particular, \cite{mckague2016interactive} and our previous paper \cite{kalev2017rigidity}).

\subsection{Related works and further directions}

In previous works on quantum rigidity,
pictures are often used
as an aid to a proof, but not as
a proof itself.  
The only other rigidity paper that we know of which used rigorous graphical methods is the recent paper \cite{coladangelo2017robust} which successfully used the concept of a group picture to prove rigidity for a new class of $2$-player games.  Group pictures are visual proofs of equations between elements of a finitely presented group.  In the context of rigidity, group pictures construct approximate relations between products of sequences of operators, and as such they are a useful general tool for proving rigidity of $2$-player games.
An interesting further direction is to try to merge the formalism of \cite{coladangelo2017robust} with the one given here in order to address general $n$-player games.

A natural next step is to explore cryptographic applications.
Since GHZ states form the basis for the secret-sharing scheme of \cite{hillery1999quantum}, it may be useful to see if the game $\overline{GHZ}_N$ can be used to create a new protocol for $3$-party secret sharing using untrusted quantum devices.

\section{Preliminaries}

\subsection{The augmented GHZ game}

The game $\overline{GHZ}_N$ that we will use
to self-test the state $\left| G \right>^{\otimes N}$ from equation (\ref{gstate}) is given in Figure~\ref{fig:centralgame}.  In this game, each player
is requested to give outputs for either one or two round numbers
(chosen from the set $\{ 1, 2, \ldots, N \}$) given
inputs for each round number.  Both the inputs and the players' outputs are expressed as partial functions on the set $\{ 1, 2, \ldots, N \}$.
This game is modeled after \cite{chao2016test}.

The variable  $r$
determines the type of input given to each player.
Note that
in the case $r = 0$, there are $4N$ possible input combinations that the referee could give to the players (since there are $N$ possible values for $i$, and $4$ possible values for
$(f_1 ( i ) , f_2 ( i ) , f_3 ( i ) )$) and for each of the values, $r = 1, 2, 3$, there are $8 N ( N - 1 )$ possible input combinations. 
Each valid input combination occurs
with probability $\Omega ( 1/N^2 )$.

\begin{figure}
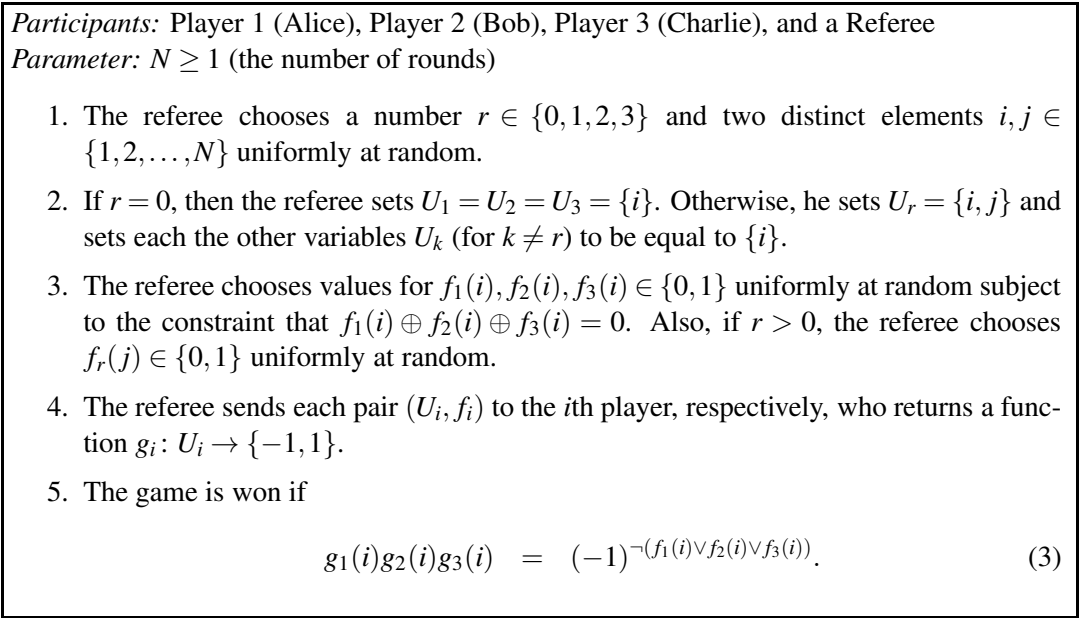

\begin{center}
\fbox{\parbox{5.5in}{
\textit{Participants:} Player 1 (Alice),
Player 2 (Bob), Player 3 (Charlie), and a Referee\\ 
\textit{Parameter:} $N \geq 1$ (the number of rounds) 

\begin{enumerate}
\item The referee chooses a number
$r \in \{ 0, 1, 2, 3 \}$ and two distinct elements $i, j \in \{ 1, 2, \ldots, N \}$ uniformly at random.

\item If $r = 0$, then the referee sets
$U_1 = U_2 = U_3 = \{ i \}$.  Otherwise,
he sets $U_r = \{ i , j \}$ and
sets each the other variables $U_k$ (for $k \neq r$) to be equal to $\{ i \}$. 

\item The referee chooses values
for $f_1(i ), f_2 ( i ), f_3 ( i) \in \{ 0, 1 \}$ uniformly at random subject to the constraint that
$f_1 (i ) \oplus f_2 ( i ) \oplus f_3 ( i ) 
 = 0$.  Also, if $r > 0$, the referee chooses $f_r ( j ) \in \{ 0, 1 \}$ uniformly at random.
 
\item The referee sends each pair $(U_i, f_i)$ to
the $i$th player, respectively, who returns
a function $g_i \colon U_i \to \{ -1 ,1 \}$.

\item The game is won if
\begin{eqnarray}
g_1 (i) g_2 ( i ) g_3 ( i ) & = & (-1)^{\neg ( f_1 ( i ) \vee f_2 ( i )
\vee f_3 ( i ))}.
\end{eqnarray}
\end{enumerate}
}
}
\end{center}
\caption{The augmented GHZ game
($\overline{GHZ}_N$).}
\label{fig:centralgame}
\end{figure}

We wish to describe the set of all possible quantum behaviors
by the players Alice, Bob, and Charlie in $\overline{GHZ}_N$.  In the definition
that follows, we use the term \textit{reflection} to mean an observable with values in $\{ \pm 1 \}$ (in other words,
a Hermitian operator whose square is the identity).
A \textit{quantum
strategy} for the game $\overline{GHZ}_N$ game
consists of the following data.
\begin{enumerate}
\item A unit vector $L \in A \otimes B \otimes C$, where
$A, B, C$ are finite-dimensional Hilbert spaces.

\item For each $i \in \{ 1, 2, \ldots, N \}$,
$b \in \{ 0, 1 \}$, and $W \in \{ A, B, C \}$, a reflection
\begin{eqnarray}
R_{i \to b}^W
\end{eqnarray}
on $W$.

\item For each $i,j \in \{ 1, 2, \ldots, N \}$,
$b,c \in \{ 0, 1 \}$, and $W \in \{ A, B, C \}$, two \textit{commuting}
reflections
\begin{eqnarray}
\label{simulmeas}
R_{i \to b \mid j \to c}^W & \textnormal{ and } & R_{ j \to c \mid i \to b}^W
\end{eqnarray}
on $W$.
\end{enumerate}

The spaces $A, B, C$ denote the registers possessed by Alice,
Bob, and Charlie, respectively.  The vector $L$ denotes
the initial state that they share before the game begins.
The reflections $R_{i \to b}^W$ describe
their behavior on singleton rounds (specifically, on a singleton round the player measures his or her register $W$
along the eigenspaces of $R_{i \to b}^W$, 
and reports either $+1$ or $-1$ for round $i$, appropriately).  The
reflections $R_{i \to b \mid j \to c}^W , R_{ j \to c \mid i \to b}^W$ describe their behavior on non-singleton rounds (specifically, if the input to a player is the function $[i \to b, j \to c]$, then 
they measure along the eigenspaces of
$R_{i \to b \mid j \to c}^W$ to obtain their
output for round $i$ and measure along the eigenspaces of $R_{ j \to c\mid  i \to b}^W$
to determine their output for round $j$).
Note that since 
the reflections in (\ref{simulmeas}) represent measurements
that are carried out simultaneously
by one of the players, we assume that these
two reflections commute.  (This assumption
will be critical in our proof).

For any reflection $Z$ and unit vector $\psi$
on a finite-dimensional Hilbert space $Q$, if we measure $\psi$ with $Z$ then the probability of obtaining an output of $-1$ is
precisely
\begin{eqnarray}
[1 - \textnormal{Tr} ( Z \psi \psi^* )]/2 & = &
\left\| Z \psi - \psi \right\|^2/4
\end{eqnarray}
We can use this fact to express the losing probabilities achieved by the players in terms of their strategy.  If $r = 0$
and the players are queried for round $i$ with
inputs $x, y,z$, then their losing probability is precisely
\begin{eqnarray}
\label{losingexp1}
\left\| R^A_{i \to x } 
R^B_{i \to y } R^C_{i \to z } L + (-1)^{x \vee y \vee z} L \right\|^2/4.
\end{eqnarray}
If Alice is queried for round $i$ with input $x$ and round $j$ with input $x'$, and Bob and Charlie are queried for round $i$ with inputs $y,z$ respectively, then the losing probability is
\begin{eqnarray}
\label{losingexp2}
\left\| R^A_{i \to x \mid j \to x'} 
R^B_{i \to y } R^C_{i \to z } L + (-1)^{x \vee y \vee z} L \right\|^2/4.
\end{eqnarray}
Similar expressions hold for the case where Bob or Charlie
is the party that receives two queries.

Note that the game $\overline{GHZ}_N$ is entirely symmetric between the three players Alice, Bob and Charlie.  This means that, given any strategy $\left( L , \left\{ R^W_{i \to b} \right\}, \left\{ R^W_{i \to b \mid j \to c } \right\} \right)$
for $\overline{GHZ}_N$,
we can produce five additional strategies by choosing a nontrivial permutation $\sigma \colon \{ 1, 2, 3 \} \to \{ 1, 2, 3 \}$ and 
giving the $p$th players' subsystem and measurement strategy to the $\sigma ( p )$th player, for each $p \in \{ 1, 2, 3 \}$.  We will make 
use of this symmetry in the proof that follows.

\subsection{Approximation chains}
\label{sec approx}

We make the following definition (see similar notation in \cite{kissinger2017picture,breiner2017graphical}).
If $F, G \colon A \to B$ are linear maps, then
\begin{eqnarray}
F & \underset{\delta}{=} & G
\end{eqnarray}
denotes the inequality
$\left\| F - G \right\|_2  \leq  {\cal O} ( \delta )$,
where $\left\| \cdot \right\|_2$ denotes
the Frobenius norm
and ${\cal O}$ denotes asymptotic big-O notation.  (If $F$ and $G$ are vectors,
then $\left\| F - G \right\|_2$
is simply the Euclidean distance
$\left\| F - G \right\|$.)

We will use this notation
especially in the case where $F$ and $G$
are processes represented by diagrams.
Note that this
relation is transitive: if $F \underset{\delta}{=} G$
and $G \underset{\delta}{=} H$, then
$F \underset{\delta}{=} H$.  Note also that if $J \colon B \to C$ is a
linear map whose operator norm satisfies $\left\| J \right\|_\infty \leq \alpha$, then $F \underset{\delta}{=} G \Longrightarrow J \circ F \underset{\delta \alpha}{=} J \circ G$.

\section{Rigidity of the augmented $GHZ$ game}

Throughout this section, suppose that
\begin{eqnarray}
\left( L , \left\{ R^A_{i \to a} \right\},
\left\{ R^B_{i \to b} \right\},
\left\{ R^C_{i \to c} \right\},
\left\{ R^A_{i \to a \mid j \to a'} \right\},
\left\{ R^B_{i \to b \mid j \to b'} \right\},
\left\{ R^C_{i \to c \mid j \to c'} \right\}
\right)
\end{eqnarray}
is a quantum strategy for the $\overline{GHZ}_N$ game
which wins with probability $1-\epsilon$.  It is helpful
to introduce some redundant notation for this strategy:
for any $W \in \{ A, B, C \}$ and $i \in \{ 1, 2, \ldots, N \}$
let
\begin{eqnarray}
X'_{W, i} & = & R^W_{i \mapsto 0}, \\
Z'_{W, i} & = & R^W_{i \mapsto 1},
\end{eqnarray}
The reason for this notation is that we intend to show that the operators 
$R^W_{i \mapsto 1}, R^W_{i \mapsto 1}$ approximate the behavior of the $X$ and $Z$ measurements in the optimal GHZ strategy (see
the beginning of subsection~\ref{statementsubsec}).
Similarly, let 
$X'_{W, i|j\to 1} =R^W_{i \to 0 \mid j \to 1 }$  and $Z'_{W, j|i\to 0} =R^W_{j \to 1 \mid i \to 0}$.

We will drop the subscript $W$ from this notation
when it is clear from the context.

\subsection{Initial steps}

Our first goal is to prove approximate commutativity
and anticommutativity relations for the operators
$X'_{W, i}, Z'_{W, i}$.  

Since the losing probability for our chosen strategy is $\epsilon$,
and each input string occurs
with probability $\Omega ( N^2 )$, we can conclude that the probability
of losing on any particular input combination is ${\cal O}( N^2 \epsilon )$.  By the discussion of expressions (\ref{losingexp1}) and (\ref{losingexp2}) above, we therefore have the following for any $i \in \{ 1, 2, \ldots, N \}$:
\begin{align}
\Bigl\| (  I + \xp_{A,i}\xp_{B,i}\xp_{C,i} ) \ket{L} \Bigr\|^2 & \leq   {\cal O}(N^2 \epsilon ),\nn
\Bigl\| (  I - \zp_{A,i}\zp_{B,i}\xp_{C,i} ) \ket{L} \Bigr\|^2 & \leq   {\cal O}(N^2 \epsilon ),\nn
\Bigl\| (  I - \xp_{A,i}\zp_{B,i}\zp_{C,i} ) \ket{L} \Bigr\|^2 & \leq   {\cal O}(N^2 \epsilon ),\nn
\Bigl\| (  I - \zp_{A,i}\xp_{B,i}\zp_{C,i} ) \ket{L} \Bigr\|^2 & \leq   {\cal O}(N^2 \epsilon ),
\label{keyineqs}
\end{align}
Additionally, if we take any of the four inequalities above, and replace any one of the operators $X_{W,i}'$ with $X'_{W, i\mid j \to t}$,
where $j \neq i$ and $t \in \{ 0, 1 \}$, the inequality remains true, and likewise if we replace any $Z_{W,i}'$ with $Z'_{W, i  \mid j \to t}$,
the inequality remains true.

We can translate the inequalities above into graphical form. 

\begin{proposition}\label{prop:switch multi}
The following inequalities hold.
\begin{eqnarray}
\begin{tikzpicture}
\node (out1) at (-1.5,2) {$A$};
\node (out2) at (0,2) {$B$};
\node (out3) at (1.5,2) {$C$};

\node[style=cprocess] (x1) at (-1.5,1.1) {$X_i'$};
\node[style=cprocess] (x2) at (0,1.1) {$X_i'$};
\node[style=cprocess] (x3) at (1.5,1.1) {$X_i'$};

\node[style=cstate, minimum width=2.5cm] (L) at (-.5,0) {$L$};

\draw (x1) to (out1);
\draw (x2) to (out2);
\draw (x3) to (out3);

\draw (x1) to (x1|-L.north);
\draw (x2) to (x2|-L.north);
\draw (x3) to (x3|-L.north);

\node (eq) at (3,.5) {\LARGE $\underset{N\sqrt{\epsilon}}{=\joinrel=}$};

\node (out4) at (4,1.5) {$A$};
\node (out5) at (5.5,1.5) {$B$};
\node (out6) at (7,1.5) {$C$};

\node[style=cstate, minimum width=2.5cm] (L2) at (5,.5) {$-L$};

\draw (out4) to (out4|-L2.north);
\draw (out5) to (out5|-L2.north);
\draw (out6) to (out6|-L2.north);

\end{tikzpicture}
\\
\label{XYYineq}
\begin{tikzpicture}
\node (out1) at (-1.5,2) {$A$};
\node (out2) at (0,2) {$B$};
\node (out3) at (1.5,2) {$C$};

\node[style=cprocess] (x1) at (-1.5,1.1) {$X_i'$};
\node[style=cprocess] (x2) at (0,1.1) {$Z_i'$};
\node[style=cprocess] (x3) at (1.5,1.1) {$Z_i'$};

\node[style=cstate, minimum width=2.5cm] (L) at (-.5,0) {$L$};

\draw (x1) to (out1);
\draw (x2) to (out2);
\draw (x3) to (out3);

\draw (x1) to (x1|-L.north);
\draw (x2) to (x2|-L.north);
\draw (x3) to (x3|-L.north);

\node (eq) at (3,.5) {\LARGE $\underset{N\sqrt{\epsilon}}{=\joinrel=}$};

\node (out4) at (4,1.5) {$A$};
\node (out5) at (5.5,1.5) {$B$};
\node (out6) at (7,1.5) {$C$};

\node[style=cstate, minimum width=2.5cm] (L2) at (5,.5) {$L$};

\draw (out4) to (out4|-L2.north);
\draw (out5) to (out5|-L2.north);
\draw (out6) to (out6|-L2.north);

\end{tikzpicture}
\end{eqnarray}
And, inequality (\ref{XYYineq}) also holds for
any permutation of the letters $A, B, C$.
\end{proposition}

We will use Proposition~\ref{prop:switch multi} to prove the following assertion.

\begin{proposition}[Approximate anti-commutativity]
\label{prop:antimulti}
For any $i$, the following inequality holds:\vspace{0.3cm}\\
\begin{equation}
\begin{tikzpicture}
\node (out1) at (-1.5,3) {$A$};
\node (out2) at (0,3) {$B$};
\node (out3) at (1.5,3) {$C$};

\node[style=cprocess] (x1) at (-1.5,1.1) {$X_i'$};
\node[style=cprocess] (y1) at (-1.5,2.1) {$Z_i'$};

\node[style=cstate, minimum width=2.5cm] (L) at (-.5,0) {$L$};

\draw (y1) to (out1);
\draw (x1) to (y1);

\draw (x1) to (x1|-L.north);
\draw (out2) to (out2|-L.north);
\draw (out3) to (out3|-L.north);

\node (eq) at (3,.5) {\LARGE $\underset{N\sqrt{\epsilon}}{=\joinrel=}$};

\node (out4) at (4,3) {$A$};
\node (out5) at (5.5,3) {$B$};
\node (out6) at (7,3) {$C$};

\node[style=cprocess] (x4) at (4,2.1) {$X_i'$};
\node[style=cprocess] (y4) at (4,1.1) {$Z_i'$};

\node[style=cstate, minimum width=2.5cm] (L) at (5,0) {$-L$};

\draw (x4) to (out4);
\draw (y4) to (x4);

\draw (y4) to (x4|-L.north);
\draw (out5) to (out5|-L.north);
\draw (out6) to (out6|-L.north);

\end{tikzpicture}
\end{equation}
\end{proposition}

\begin{proof}
Repeatedly applying Proposition~\ref{prop:switch multi},
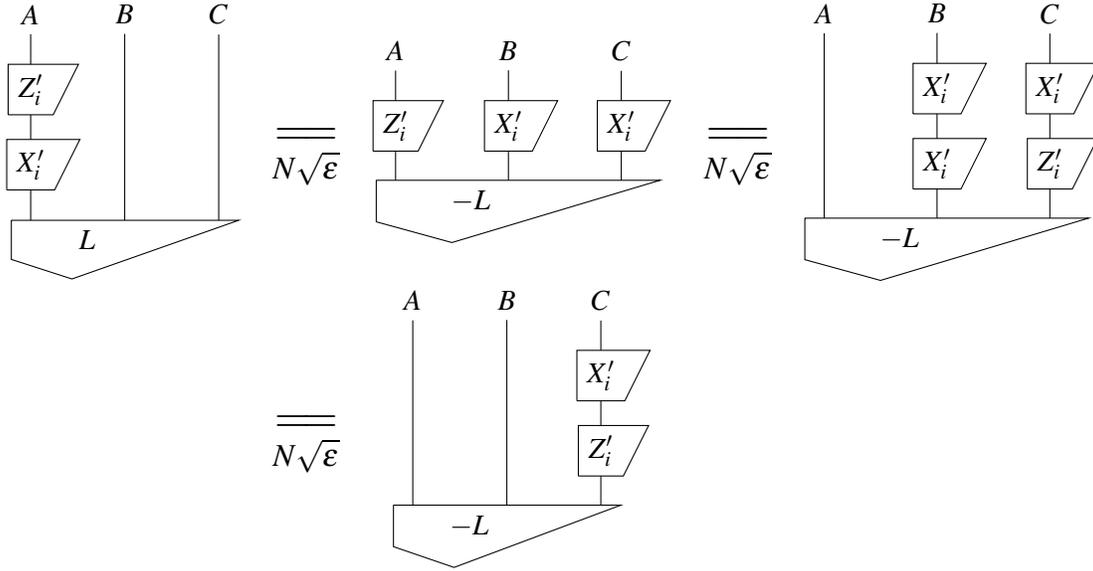
\begin{equation}
\begin{array}{ccccc}
\vcenter{\hbox{
\begin{tikzpicture}
\node (out1) at (-1.25,3) {$A$};
\node (out2) at (0,3) {$B$};
\node (out3) at (1.25,3) {$C$};

\node[style=cprocess] (x1) at (-1.25,2) {$Z_{i}'$};

\node[style=cprocess] (y1) at (-1.25,1) {$X_{i}'$};

\node[style=cstate, minimum width=2cm] (L) at (-.5,0) {$L$};

\draw (out1) to (x1) to (y1) to (y1|-L.north);

\draw (out2) to (out2|-L.north);

\draw (out3) to (out3|-L.north);
\end{tikzpicture}
}}
& \mbox{\LARGE $\underset{N\sqrt{\epsilon}}{=\joinrel=}$} &
\vcenter{\hbox{
\begin{tikzpicture}
\node (out1) at (-1.5,2) {$A$};
\node (out2) at (0,2) {$B$};
\node (out3) at (1.5,2) {$C$};


\node[style=cprocess] (y1) at (-1.5,1) {$Z_{i}'$};
\node[style=cprocess] (y2) at (0,1) {$X_{i}'$};
\node[style=cprocess] (y3) at (1.5,1) {$X_{i}'$};

\node[style=cstate, minimum width=2.5cm] (L) at (-.5,0) {$-L$};

\draw (out1) to (y1) to (y1|-L.north);

\draw (out2) to (y2) to (y2|-L.north);

\draw (out3) to (y3) to (y3|-L.north);
\end{tikzpicture}
}}
& \mbox{\LARGE $\underset{N\sqrt{\epsilon}}{=\joinrel=}$} &
\vcenter{\hbox{
\begin{tikzpicture}
\node (out1) at (-1.5,3) {$A$};
\node (out2) at (0,3) {$B$};
\node (out3) at (1.5,3) {$C$};

\node[style=cprocess] (x2) at (0,2) {$X_{i}'$};
\node[style=cprocess] (x3) at (1.5,2) {$X_{i}'$};

\node[style=cprocess] (y2) at (0,1) {$X_{i}'$};
\node[style=cprocess] (y3) at (1.5,1) {$Z_{i}'$};

\node[style=cstate, minimum width=2.5cm] (L) at (-.5,0) {$-L$};

\draw (out1) to (out1|-L.north);

\draw (out2) to (x2) to (y2) to (y2|-L.north);

\draw (out3) to (x3) to (y3) to (y3|-L.north);
\end{tikzpicture}
}} 
\\
& \mbox{\LARGE $\underset{N\sqrt{\epsilon}}{=\joinrel=}$} &
\vcenter{\hbox{
\begin{tikzpicture}
\node (out1) at (-1.25,3) {$A$};
\node (out2) at (0,3) {$B$};
\node (out3) at (1.25,3) {$C$};

\node[style=cprocess] (x3) at (1.25,2) {$X_{i}'$};

\node[style=cprocess] (y3) at (1.25,1) {$Z_{i}'$};

\node[style=cstate, minimum width=2cm] (L) at (-.5,0) {$-L$};

\draw (out1) to (out1|-L.north);

\draw (out2) to (out2|-L.north);

\draw (out3) to (x3) to (y3) to (y3|-L.north);
\end{tikzpicture}
}}
\end{array}
\end{equation}
(Here we have
used the fact that all of the unitary maps in these diagrams are self-inverse.)  Applying the same steps symmetrically across the wires $A, B, C$,
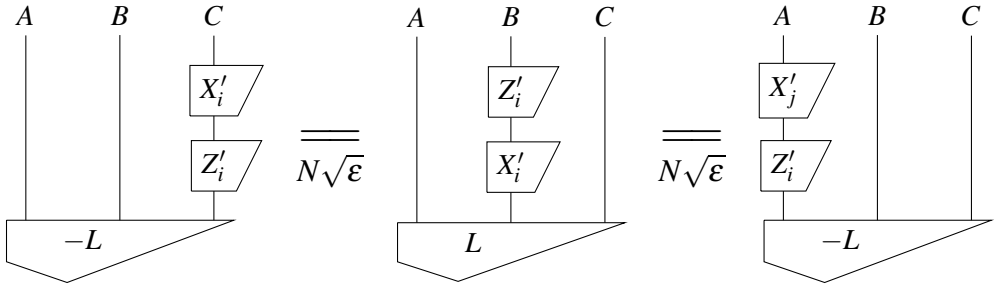
\begin{equation}
\begin{array}{ccccc}
\vcenter{\hbox{
\begin{tikzpicture}
\node (out1) at (-1.25,3) {$A$};
\node (out2) at (0,3) {$B$};
\node (out3) at (1.25,3) {$C$};

\node[style=cprocess] (x3) at (1.25,2) {$X_{i}'$};

\node[style=cprocess] (y3) at (1.25,1) {$Z_{i}'$};

\node[style=cstate, minimum width=2cm] (L) at (-.5,0) {$-L$};

\draw (out1) to (out1|-L.north);

\draw (out2) to (out2|-L.north);

\draw (out3) to (x3) to (y3) to (y3|-L.north);
\end{tikzpicture}
}} & \mbox{\LARGE $\underset{N\sqrt{\epsilon}}{=\joinrel=}$} &
\vcenter{\hbox{
\begin{tikzpicture}
\node (out1) at (-1.25,3) {$A$};
\node (out2) at (0,3) {$B$};
\node (out3) at (1.25,3) {$C$};

\node[style=cprocess] (x2) at (0,2) {$Z_{i}'$};

\node[style=cprocess] (y2) at (0,1) {$X_{i}'$};

\node[style=cstate, minimum width=2cm] (L) at (-.5,0) {$L$};

\draw (out1) to (out1|-L.north);

\draw (out2) to (x2) to (y2) to (y2|-L.north);

\draw (out3) to (out3|-L.north);
\end{tikzpicture}
}} & \mbox{\LARGE $\underset{N\sqrt{\epsilon}}{=\joinrel=}$} &
\vcenter{\hbox{
\begin{tikzpicture}
\node (out1) at (-1.25,3) {$A$};
\node (out2) at (0,3) {$B$};
\node (out3) at (1.25,3) {$C$};

\node[style=cprocess] (x1) at (-1.25,2) {$X_{j}'$};

\node[style=cprocess] (y1) at (-1.25,1) {$Z_{i}'$};

\node[style=cstate, minimum width=2cm] (L) at (-.5,0) {$-L$};

\draw (out1) to (x1) to (y1) to (y1|-L.north);

\draw (out2) to (out2|-L.north);

\draw (out3) to (out3|-L.north);
\end{tikzpicture}
}}
\end{array}
\end{equation}
as desired.
\end{proof}


\begin{proposition}[Approximate commutativity]
\label{prop:app commutativity}
The following equation holds for any distinct indices $i, j \in \{ 1, 2, \ldots, N \}$ and any bits $b,c\in\{0,1\}$:\vspace{0.3cm}\\
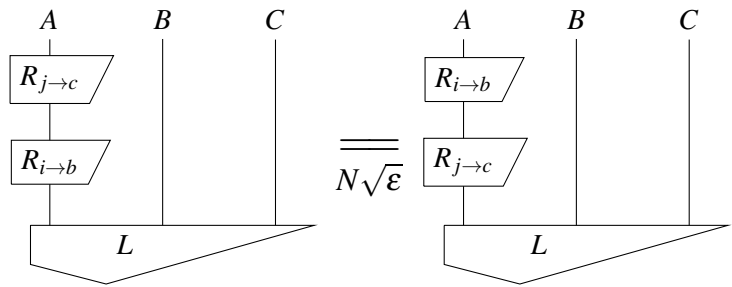
\begin{center}
\begin{tikzpicture}
\node (out1) at (-1.5,3) {$A$};
\node (out2) at (0,3) {$B$};
\node (out3) at (1.5,3) {$C$};

\node[style=cprocess] (j1) at (-1.5,2.2) {$R_{j \to c}$};
\node[style=cprocess] (i1) at (-1.5,1.1) {$R_{i \to b}$};

\node[style=cstate, minimum width=2.5cm] (L) at (-.5,0) {$L$};

\draw (j1) to (out1);
\draw (i1) to (j1);

\draw (i1) to (i1|-L.north);
\draw (out2) to (out2|-L.north);
\draw (out3) to (out3|-L.north);

\node (eq) at (2.75,1) {\LARGE $\underset{N \sqrt{\epsilon}}{=\joinrel=}$};

\node (out4) at (4,3) {$A$};
\node (out5) at (5.5,3) {$B$};
\node (out6) at (7,3) {$C$};

\node[style=cprocess] (j2) at (4,1.1) {$R_{j \to c}$};
\node[style=cprocess] (i2) at (4,2.2) {$R_{i \to b}$};

\node[style=cstate, minimum width=2.5cm] (L2) at (5,0) {$L$};

\draw (i2) to (out4);
\draw (j2) to (i2);

\draw (j2) to (j2|-L.north);
\draw (out5) to (out5|-L.north);
\draw (out6) to (out6|-L.north);
\end{tikzpicture}
\end{center}

\end{proposition}

The proof of Proposition~\ref{prop:app commutativity} is given 
in appendix~\ref{appendix comm proofs}, and is based
on the fact that the related
reflections $R_{i \to b
\mid j \to c}$ and $R_{j \to c \mid i \to b}$ commute.

\subsection{Isometries}

Next we define the local isometries which will relate the state $L$ to the ideal state
$\left| GHZ \right>^{\otimes N}$.  Roughly speaking, for each
$k \in \{ 1, 2, \ldots, N \}$
and $W \in \{ A, B, C \}$, we
will construct an isometry $\Psi_{W, k}$ which approximately ``locates'' a qubit with the $i$th players' system, and swaps it out onto a qubit register $Q \cong \mathbb{C}^2$.  We then apply these isometries in order, for $1, 2, \ldots, N$, the system
$W$.  We will use approximate commutativity to show that the different isometries $\Psi_{W, k}$ do not interfere much with one another when applied in sequence.
Our approach borrows from previous works on rigidity, and uses similar notation to that of the non-graphical rigidity proof in our previous paper \cite{kalev2017rigidity}.

In the following, we use  controlled unitary gates, $\mathbf{C}(U)$, which are defined and discussed in Appendix \ref{appendix cu}.  
Let $H \colon 
\mathbb{C}^2 \to \mathbb{C}^2$ denote the Hadamard gate 
$[ \left| 0 \right> \mapsto
\left| + \right>, \left| 1 \right>
\mapsto \left| - \right>]$.  
We define isometries 
$\Psi_{A, k}$ on the system
$A$ which involve preparing
a Bell state (denoted by a gray node) and then
performing an approximate ``swap'' procedure
between one half of the Bell state and $A$.
(This is based on \cite{mckague2016interactive}.)  Then we define an isometry $\Theta_{A,k}$ on $A$ which chains together the swapping maps $\Psi_{A, 1}, \ldots, \Psi_{A,k}$.

\begin{definition}[Swapping maps]
\label{def:swapmap}
For each $k \in \{ 1, 2, \ldots, N \}$, let $Q_k$ and $\overline{Q}_k$ denote qubit
registers, and define an isometry $\Psi_{A,k}$ as follows (suppressing the label $A$ when it is not necessary):

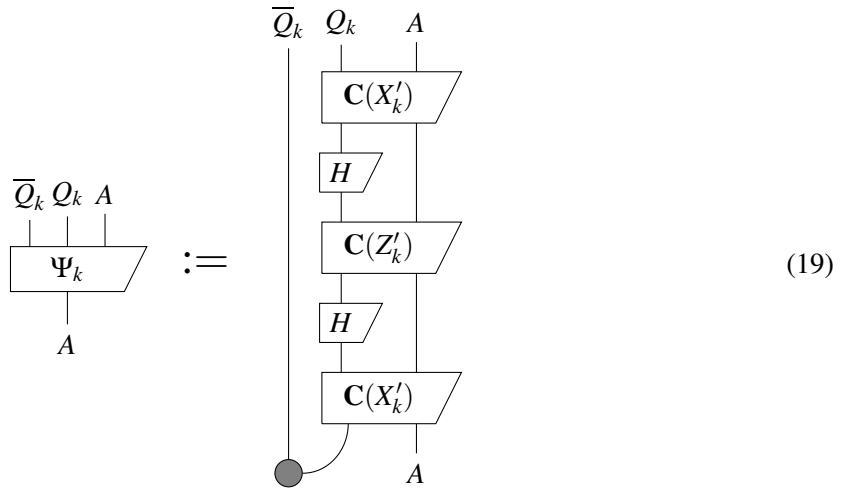
\begin{equation}
\label{swapmap}
\begin{array}{ccc}
\vcenter{\hbox{
\begin{tikzpicture}

\node (in) at (0,-1) {$A$};

\node (out1) at (-.5,1) {$\overline{Q}_k$};
\node (out3) at (0,1) {$Q_k$};
\node (out2) at (.5,1) {$A$};

\node[style=cprocess,minimum width = 1.5cm] (psi) at (0,0) {$\Psi_{k}$};

\draw (in) to (in|-psi.south);

\draw (out1) to (out1|-psi.north);
\draw (out2) to (out2|-psi.north);
\draw (out3) to (out3|-psi.north);

\end{tikzpicture}
}}
& \mbox{\LARGE $:=$} &
\vcenter{\hbox{
\begin{tikzpicture}
\node (in) at (.5,-2.5) {$A$};

\node[style=bell] (bell) at (-1.2,-2.5) {};

\node (out1) at (-.5,3.5) {$Q_k$};
\node (out2) at (.5,3.5) {$A$};

\node[style=cprocess,minimum width = 1.5cm] (CX) at (0,0.5) {$\mathbf{C}(Z_k')$};
\node[style=cprocess] (S) at (-.5,1.5) {$H$};
\node[style=cprocess,minimum width = 1.5cm] (CY) at (0,2.5) {$\mathbf{C}(X_k')$};

\node[style=cprocess] (S2) at (-.5,-0.5) {$H$};

\node[style=cprocess,minimum width = 1.5cm] (CZ) at (0,-1.5) {$\mathbf{C}(X_k')$};

\draw (in) to (in|-CZ.south);

\draw (in|-CX.north) to (in|-CY.south);
\draw (out2) to (out2|-CY.north);

\draw (in|-CZ.north) to (in|-CX.south);

\draw (bell) to[in=-90, out=0] (CZ.220);
\draw (S) to (S|-CX.north);
\draw (S) to (S|-CY.south);
\draw (out1) to (out1|-CY.north);

\draw (S2) to (S2|-CX.south);
\draw (S2) to (S2|-CZ.north);

\node(out3) at (-1.2,3.5) {$\overline{Q}_k$};

\draw (bell) to (out3);

\node(empty) at (0,-3.2) {};

\end{tikzpicture}
}} 
\end{array}
\end{equation}

Let $\mathbf{Q}_k = \overline{Q}_k \otimes 
Q_k$ and $\mathbf{Q}_{1 \ldots k}
= \mathbf{Q}_1 \otimes \cdots
\otimes \mathbf{Q}_k$. 
 Define
an isometry $\Theta_{k, A}$ by
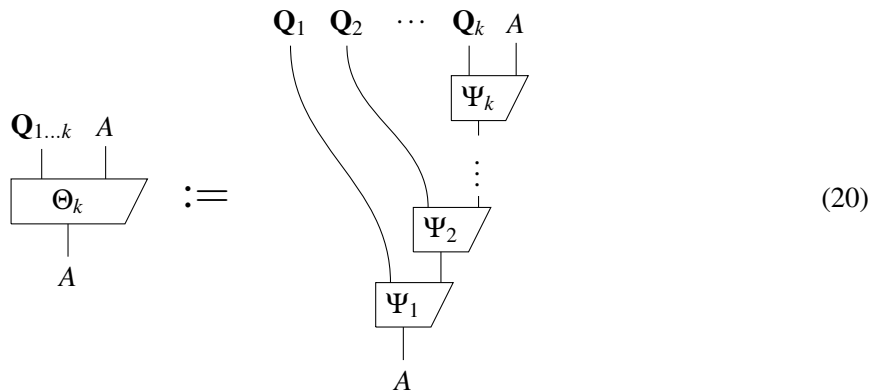
\begin{equation}
\begin{array}{ccc}
\vcenter{\hbox{
\begin{tikzpicture}

\node (out1) at (-.35,1) {$\mathbf{Q}_{1 \ldots k }$};
\node (out2) at (.5,1) {$A$};

\node[style=cprocess,minimum width=1.5cm] (psi) at (0,0) {$\Theta_{k}$};

\node (in) at (0,-1) {$A$};

\draw (out1) to (out1|-psi.north);
\draw (out2) to (out2|-psi.north);

\draw (in) to (in|-psi.south);

\end{tikzpicture}
}}
& \mbox{\LARGE $:=$} &
\vcenter{\hbox{
\begin{tikzpicture}

\node (in) at (0,-1) {$A$};

\node[style=cprocess] (psi1) at (0,0) {$\Psi_{1}$};
\node[style=cprocess] (psi2) at (.5,1) {$\Psi_{2}$};
\node[style=cprocess] (psi3) at (1,2.75) {$\Psi_{k}$};

\node (out1) at (-1.5,3.75) {$\mathbf{Q}_1$};
\node (out2) at (-.75,3.75) {$\mathbf{Q}_2$};
\node (out3) at (.875,3.75) {$\mathbf{Q}_k$};
\node (out4) at (1.5,3.75) {$A$};

\node (dots1) at (.125,3.75) {$\cdots$};
\node (dots2) at (1,1.85) {$\vdots$};

\draw (in) to (psi1);
\draw (psi2) to (psi2|-psi1.north);
\draw (dots2) to (psi3|-psi2.north);
\draw (psi3) to (psi3|-dots2.north);

\draw (out4) to (out4|-psi3.north);
\draw (out3) to (out3|-psi3.north);

\draw (out2) to[in=90,out=-90] (psi2.120);
\draw (out1) to[in=90,out=-90] (psi1.120);
\end{tikzpicture}
}}
\end{array}
\end{equation}
Let $Q_{N+1}, \ldots, Q_{3N}, \overline{Q}_{N+1}, \ldots, \overline{Q}_{3N}$ be qubit registers, and
define $\Psi_{B, k}$ analogously as an isometry from $B$ to $B \otimes \overline{Q}_{N + k} \otimes 
Q_{N+k}$.  Define
$\Psi_{C, k}$ analogously as an isometry
from $C$ to $C \otimes \overline{Q}_{2N + k} \otimes 
Q_{2N+k}$.  Define composite maps
$\Theta_{B, k}$ 
and $\Theta_{C,k}$
similarly in terms of $\Psi_{B, k}$
and $\Psi_{C, k}$.
\end{definition}

\subsection{Commutativity properties}
\label{sec:comm}

We now investigate some approximate (anti-)commutativity relationships between the Pauli operators on $Q$, the reflection strategies for $A$, $B$ and $C$, and the isometries defined in the last section.

We begin with the following definition and lemma, which are crucial.

\begin{definition}
Let $R, S$ be registers and let
$Z \in R \otimes S$ be a unit vector.  If a unitary map $U \colon R \to R$ is such that
there exists another unitary map $V \colon S \to S$ satisfying
\begin{equation}
\begin{array}{ccc}
\vcenter{\hbox{
\begin{tikzpicture}
\node[style=cstate] (z) at
(0,0) {$~~~~Z~~~~~$};
\node[style=cprocess] (u)
at (-0.3,1) {$U$};
\draw (u) to (u|-z.north);
\node (r) at (-0.3,2) {$R$};
\node (s) at (0.5,2) {$S$};
\draw (s) to (s|-z.north);
\draw (r) to (r|-u.north);
\end{tikzpicture}}}
& \vcenter{\hbox{{\LARGE $\underset{ \delta}{=\joinrel=}$}}} & 
\vcenter{\hbox{
\begin{tikzpicture}
\node[style=cstate] (z) at
(0,0) {$~~~~Z~~~~~$};
\node[style=cprocess] (v)
at (0.5,1) {$V$};
\draw (r) to (r|-z.north);
\node (r) at (-0.3,2) {$R$};
\node (s) at (0.5,2) {$S$};
\draw (s) to (s|-v.north);
\draw (v) to (v|-z.north);
\end{tikzpicture}}}
\end{array}
\end{equation}
the we say that \textnormal{$U$
can be pushed through $Z$ with
error term $\delta$}.
\end{definition}

\begin{lemma}[Push Lemma]
\label{lemma:push}
Suppose that $R, S$ are registers,
$Z \in R \otimes S$ is a unit vector, and $V, W, U_1, U_2, \ldots, U_k$ are unitary operators
on $R$ such that
\begin{enumerate}
\item Each map $U_i$ can be pushed
through $Z$ with error term $\epsilon$, and
\item The approximate
equality
$(V \otimes I_S)L \underset{\delta}{=} (W \otimes I_S)L$ holds.
\end{enumerate}
Then,
\begin{equation}
\begin{array}{ccc}
\vcenter{\hbox{
\begin{tikzpicture}
\node[style=cstate] (z) at
(0,0) {$~~~~Z~~~~~$};
\node[style=cprocess] (uk)
at (-0.3,1) {$U_k$};
\draw (uk) to (uk|-z.north);
\node[style=cprocess] (u1)
at (-0.3,3) {$U_1$};
\draw (uk) to (uk|-z.north);
\node (dots) at (-0.3,2) {$\vdots$};
\draw (u1) to (u1|-dots.north);
\draw (dots) to (dots|-uk.north);
\node[style=cprocess] (v) at (-0.3,4) {$V$};
\draw (v) to (v|-u1.north);
\node (r) at (-0.3,5) {$R$};
\node (s) at (0.5,5) {$S$};
\draw (s) to (s|-z.north);
\draw (r) to (r|-v.north);
\end{tikzpicture}}}
& \vcenter{\hbox{{\LARGE $\underset{k\epsilon + \delta}{=\joinrel=}$}}} & 
\vcenter{\hbox{
\begin{tikzpicture}
\node[style=cstate] (z) at
(0,0) {$~~~~Z~~~~~$};
\node[style=cprocess] (uk)
at (-0.3,1) {$U_k$};
\draw (uk) to (uk|-z.north);
\node[style=cprocess] (u1)
at (-0.3,3) {$U_1$};
\draw (uk) to (uk|-z.north);
\node (dots) at (-0.3,2) {$\vdots$};
\draw (u1) to (u1|-dots.north);
\draw (dots) to (dots|-uk.north);
\node[style=cprocess] (w) at (-0.3,4) {$W$};
\draw (w) to (w|-u1.north);
\node (r) at (-0.3,5) {$R$};
\node (s) at (0.5,5) {$S$};
\draw (s) to (s|-z.north);
\draw (r) to (r|-v.north);
\end{tikzpicture}}}
\end{array}
\end{equation}
\end{lemma}

The Push Lemma follows from
an easy inductive argument, and is given in the appendix.
Note that by Proposition~\ref{prop:switch multi}, for any $k$ we have
\begin{eqnarray}
X'_{A, k} L & \underset{N \sqrt{\epsilon}}{=\joinrel=} & (-X'_{B, k} \otimes X'_{C, k}) L \\
Z'_{A, k} L & \underset{N \sqrt{\epsilon}}{=\joinrel=} & (Z'_{B, k} \otimes X'_{C, k}) L,
\end{eqnarray}
and so all of the maps $X'_{\cdot, k}$ and
$Z'_{\cdot k}$ can be pushed through $L$
with error term $ N \sqrt{\epsilon}$.  This fact underlies the proofs of the next two results, which are proved in the appendix
using a combination of the Push Lemma,
and the approximate commutativity and anti-commutativity properties of maps $X'_{\cdot k}$ and $Z'_{\cdot k}$. 

\begin{proposition}
\label{prop:correct pauli}
For $k \in \{ 1, 2, \ldots, N \}$, let
$X_{A,k}$ and $Z_{A,k}$ denote the Pauli
operators on $Q_k$.  Then,
\begin{equation}
\begin{array}{ccccc}
\vcenter{\hbox{
\begin{tikzpicture}

\node (out1) at (-1.5,3) {$\mathbf{Q}_{k}$};
\node (out2) at (-.5,3) {$A$};
\node (out3) at (.5,3) {$B$};
\node (out4) at (1.5,3) {$C$};

\node[style=cprocess] (X) at (-1.5,2) {$X_{k}$}; 

\node[style=cprocess, minimum width =1.5cm] (psi) at (-1,1) {$\Psi_{k}$};

\node[style=cstate,minimum width=2.5cm] (L) at (-.5,0) {$L$};

\draw (out1) to[out=-90,in=90] (out1|-X.north);
\draw (X) to (X|-psi.north);
\draw (out2) to (out2|-psi.north);
\draw (out3) to (out3|-L.north);
\draw (out4) to (out4|-L.north);

\draw (psi) to (psi|-L.north);

\end{tikzpicture}
}}
& \mbox{\LARGE $\underset{N\sqrt{\epsilon}}{=\joinrel=}$} &
\vcenter{\hbox{
\begin{tikzpicture}

\node (out1) at (-1.5,3) {$\mathbf{Q}_k$};
\node (out2) at (-.5,3) {$A$};
\node (out3) at (.5,3) {$B$};
\node (out4) at (1.5,3) {$C$};

\node[style=cprocess, minimum width =1.5cm] (psi) at (-1,2) {$\Psi_{k}$};

\node[style=cprocess] (X) at (-1,1) {$X'_{k}$}; 

\node[style=cstate,minimum width=2.5cm] (L) at (-.5,0) {$L$};

\draw (out1) to[out=-90,in=90] (out1|-psi.north);
\draw (out2) to (out2|-psi.north);
\draw (out3) to (out3|-L.north);
\draw (out4) to (out4|-L.north);

\draw (psi) to (X) to (psi|-L.north);

\end{tikzpicture}}}
\end{array}
\end{equation}
and similarly for $Z'_{k}$.  Likewise,
define $\{ X_{B, k} , Z_{B, k} \}$ 
to be the Pauli operators
on $Q_{N+k}$ and 
define $\{ X_{C, k}, Z_{C, k} \}$ to
be the Pauli operators on $Q_{2N+k}$.  Analogous statements
hold for $\Psi_{B, k}$ and $\Psi_{C, k}$.
\end{proposition}

\begin{proposition}
\label{prop:multi pauli}
For any $k \in \{ 1, 2, \ldots, N \}$, 
\begin{equation}
\begin{array}{ccccc}
\vcenter{\hbox{
\begin{tikzpicture}

\node (out1) at (-1.5,3) {$\mathbf{Q}_{1
\ldots N}$};
\node (out2) at (-.5,3) {$A$};
\node (out3) at (.5,3) {$B$};
\node (out4) at (1.5,3) {$C$};

\node[style=cprocess] (X) at (-1.5,2) {$X_{k}$}; 

\node[style=cprocess, minimum width =1.5cm] (psi) at (-1,1) {$\Theta_N$};

\node[style=cstate,minimum width=2.5cm] (L) at (-.5,0) {$L$};

\draw (out1) to[out=-90,in=90] (out1|-X.north);
\draw (X) to (X|-psi.north);
\draw (out2) to (out2|-psi.north);
\draw (out3) to (out3|-L.north);
\draw (out4) to (out4|-L.north);

\draw (psi) to (psi|-L.north);

\end{tikzpicture}
}}
& \mbox{\LARGE $\underset{N^3 \sqrt{\epsilon}}{=\joinrel=}$} &
\vcenter{\hbox{
\begin{tikzpicture}

\node (out1) at (-1.5,3) {$\mathbf{Q}_{1 \ldots N}$};
\node (out2) at (-.5,3) {$A$};
\node (out3) at (.5,3) {$B$};
\node (out4) at (1.5,3) {$C$};

\node[style=cprocess, minimum width =1.5cm] (psi) at (-1,2) {$\Theta_N$};

\node[style=cprocess] (X) at (-1,1) {$X'_{k}$}; 

\node[style=cstate,minimum width=2.5cm] (L) at (-.5,0) {$L$};

\draw (out1) to[out=-90,in=90] (out1|-psi.north);
\draw (out2) to (out2|-psi.north);
\draw (out3) to (out3|-L.north);
\draw (out4) to (out4|-L.north);

\draw (psi) to (X) to (psi|-L.north);

\end{tikzpicture}}}
\end{array}
\end{equation}
and similarly for $Z'_{k}$.  Analogous statements
hold for $\Theta_{B, k}$ and $\Theta_{C, k}$.
\end{proposition}

\subsection{Rigidity}

We are now ready to
state and prove our main result.

\begin{proposition}[Rigidity]
\label{prop:state-approx multi}
Let $\Theta_{A,N}, \Theta_{B, N},
\Theta_{C,N}$ be the isometries from
Definition~\ref{def:swapmap}.  Then, there is some state $L'$ on $A\otimes B\otimes C \otimes \overline{Q}_{1 \ldots 3N}$ such that
\begin{equation}
\begin{array}{ccccc}
\vcenter{\hbox{
\begin{tikzpicture}
\node (out1) at (-2.5,3) {\small $ \mathbf{Q}_{1 \ldots N}$};
\node (out2) at (-0.5,3) {$\mathbf{Q}_{N+1 \ldots 2N}$};
\node (out3) at (1.5,3) {$\mathbf{Q}_{2N+1 \ldots 3N}$};

\node (out4) at (-2,2.5) {$A$};
\node (out5) at (-0,2.5) {$B$};
\node (out6) at (2.5,2.5) {$C$};

\node[style=cprocess] (psi1) at (-1.75,1) {$\Theta_N$};
\node[style=cprocess] (psi2) at (0,1) {$\Theta_N$};
\node[style=cprocess] (psi3) at (1.75,1) {$\Theta_N$};

\node[style=cstate,minimum width = 3.5cm] (L) at (-.75,0) {$L$};

\draw (out1) to[out=-90,in=90] (psi1.120);
\draw (out2) to[out=-90,in=90] (psi2.120);
\draw (out3) to[out=-90,in=90] (psi3.120);

\draw (out4) to[out=-90,in=90] (psi1.60);
\draw (out5) to[out=-90,in=90] (psi2.60);
\draw (out6) to[out=-90,in=90] (psi3.60);

\draw (psi1) to (psi1|-L.north);
\draw (psi2) to (psi2|-L.north);
\draw (psi3) to (psi3|-L.north);

\end{tikzpicture}
}}
& \mbox{\LARGE $\underset{N^4 \sqrt{\epsilon}}{=\joinrel=}$} &
\vcenter{\hbox{
\begin{tikzpicture}

\node (out1) at (-2.75,3.25) {$Q_{1 \ldots N}$};
\node (out2) at (-1.6, 2.75) {$Q_{N+1 \ldots 2N}$};
\node (out3) at (-.5,3.25) {$Q_{2N+1 \ldots 3N}$};

\node (out4) at (.5,2.75) {$A$};
\node (out5) at (1.3,2.75) {$B$};
\node (out6) at (2.1,2.75) {$C$};
\node (out7) at (3.2, 2.75) {$\overline{Q}_{1 \ldots 3N}$};

\node[style=cstate,minimum width = 1.5cm] (GHZ) at (-1.5,0) {$G^{\otimes N}$};
\node[style=cstate,minimum width = 1.5cm] (L) at (1.5,0) {$L'$};

\draw (out1) to[out=-90,in=90] (GHZ.150);
\draw (out2) to[out=-90,in=90] (GHZ.80);
\draw (out3) to[out=-90,in=90] (GHZ.20);

\draw (out4) to[out=-90,in=90] (L.150);
\draw (out5) to[out=-90,in=90] (L.80);
\draw (out6) to[out=-90,in=90] (L.40);
\draw (out7) to[out=-90,in=90] (L.20);
\end{tikzpicture}
}}\\
\end{array}
\end{equation}
\end{proposition}

\begin{proof}
In the following, we write
$\mathbf{Q}^1 := \mathbf{Q}_{1 \ldots N}$,
$\mathbf{Q}^2 := \mathbf{Q}_{(N+1) \ldots 2N}$
and $\mathbf{Q^3} := \mathbf{Q}_{(2N+1)
\ldots 3N}$ in order to conserve space.
By application of Props.~\ref{prop:switch multi} and~\ref{prop:multi pauli} we have,\vspace{0.3cm}\\

\begin{equation}\label{prop:rigidity XXX}
\scalebox{.85}{
$\begin{array}{ccccc}
\vcenter{\hbox{
\begin{tikzpicture}

\node (out1) at (-2.5,3.75) {$\mathbf{Q}^1$};
\node (out2) at (-1.5,3.75) {$\mathbf{Q}^2$};
\node (out3) at (-.5,3.75) {$\mathbf{Q}^3$};

\node (out4) at (.5,3.75) {$A$};

\node (out5) at (1.5,3.75) {$B$};
\node (out6) at (2.5,3.75) {$C$};

\node[style=cprocess] (x1) at (-2.5,2) {$ X_i$};
\node[style=cprocess] (x2) at (-.5,2) {$X_i$};
\node[style=cprocess] (x3) at (1.5,2) {$X_i$};

\node[style=cprocess,minimum width=1.25cm] (psi1) at (-2,1) {$\Theta_N$};
\node (psi1mid) at (-1.5,2) {};
\node[style=cprocess,minimum width=1.25cm] (psi2) at (0,1) {$\Theta_N$};
\node (psi2mid) at (.5,2) {};
\node[style=cprocess,minimum width=1.25cm] (psi3) at (2,1) {$\Theta_N$};

\node[style=cstate,minimum width = 3.5cm] (L) at (-.75,0) {$L$};

\draw (out1) to[out=-90,in=90] (x1);
\draw (out2) to[out=-90,in=90] (x2);
\draw (out3) to[out=-90,in=90] (x3);

\draw (x1) to (x1|-psi1.north);
\draw (x2) to (x2|-psi2.north);
\draw (x3) to (x3|-psi3.north);

\draw (out4) to[out=-90,in=90] (psi1mid) to (psi1mid|-psi1.north);
\draw (out5) to[out=-90,in=90] (psi2mid) to (psi2mid|-psi2.north);
\draw (out6) to[out=-90,in=90] (out6|-psi3.north);

\draw (psi1) to (psi1|-L.north);
\draw (psi2) to (psi2|-L.north);
\draw (psi3) to (psi3|-L.north);

\end{tikzpicture}
}}
& \mbox{\LARGE $\underset{N^3\sqrt{\epsilon}}{=\joinrel=}$} &
\vcenter{\hbox{
\begin{tikzpicture}

\node (out1) at (-2.5,3.75) {$\mathbf{Q}^1$};
\node (out2) at (-1.5,3.75) {$\mathbf{Q}^2$};
\node (out3) at (-.5,3.75) {$\mathbf{Q}^3$};

\node (out4) at (.5,3.75) {$A$};
\node (out5) at (1.5,3.75) {$B$};
\node (out6) at (2.5,3.75) {$C$};

\node[style=cprocess,minimum width=1.25cm] (psi1) at (-2,2) {$\Theta_N$};
\node[style=cprocess,minimum width=1.25cm] (psi2) at (0,2) {$\Theta_N$};
\node[style=cprocess,minimum width=1.25cm] (psi3) at (2,2) {$\Theta_N$};

\node[style=cprocess] (x1) at (-2,1) {$X_{i}'$};
\node[style=cprocess] (x2) at (0,1) {$X_{i}'$};
\node[style=cprocess] (x3) at (2,1) {$X_{i}'$};

\node[style=cstate,minimum width = 3.5cm] (L) at (-.75,0) {$L$};

\draw (out1) to[out=-90,in=90] (psi1.135);
\draw (out2) to[out=-90,in=90] (psi2.135);
\draw (out3) to[out=-90,in=90] (psi3.135);

\draw (psi1) to (x1) to (x1|-L.north);
\draw (psi2) to (x2) to (x2|-L.north);
\draw (psi3) to (x3) to (x3|-L.north);

\draw (out4) to[out=-90,in=90] (psi1.45);
\draw (out5) to[out=-90,in=90] (psi2.45);
\draw (out6) to[out=-90,in=90] (psi3.45);

\end{tikzpicture}
}}
\\\\
& \mbox{\LARGE $\underset{N\sqrt{\epsilon}}{=\joinrel=}$} &
\vcenter{\hbox{
\begin{tikzpicture}

\node (out1) at (-2.5,2.75) {$\mathbf{Q}^1$};
\node (out2) at (-1.5,2.75) {$\mathbf{Q}^2$};
\node (out3) at (-.5,2.75) {$\mathbf{Q}^3$};

\node (out4) at (.5,2.75) {$A$};
\node (out5) at (1.5,2.75) {$B$};
\node (out6) at (2.5,2.75) {$C$};

\node[style=cprocess,minimum width=1.25cm] (psi1) at (-2,1) {$\Theta_N$};
\node[style=cprocess,minimum width=1.25cm] (psi2) at (0,1) {$\Theta_N$};
\node[style=cprocess,minimum width=1.25cm] (psi3) at (2,1) {$\Theta_N$};

\node[style=cstate,minimum width = 3.5cm] (L) at (-.75,0) {$-L$};

\draw (out1) to[out=-90,in=90] (psi1.135);
\draw (out2) to[out=-90,in=90] (psi2.135);
\draw (out3) to[out=-90,in=90] (psi3.135);

\draw (psi1) to (psi1|-L.north);
\draw (psi2) to (psi2|-L.north);
\draw (psi3) to (psi3|-L.north);

\draw (out4) to[out=-90,in=90] (psi1.45);
\draw (out5) to[out=-90,in=90] (psi2.45);
\draw (out6) to[out=-90,in=90] (psi3.45);

\end{tikzpicture}}}
\end{array}$}\end{equation}

(Here, $X_i$ is used to denote the $X$-Pauli
operator on either $Q_i$, $Q_{N+i}$,
or $Q_{2N+i}$ depending on which 
wire it is applied to.)
Similarly we obtain that\vspace{-0.1cm}\\

\begin{equation}\label{prop:rigidity XYY}
\scalebox{.85}{
$\begin{array}{ccccc}
\vcenter{\hbox{
\begin{tikzpicture}

\node (out1) at (-2.5,3.75) {$\mathbf{Q}^{1}$};
\node (out2) at (-1.5,3.75) {$\mathbf{Q}^{2}$};
\node (out3) at (-.5,3.75) {$\mathbf{Q}^{3}$};

\node (out4) at (.5,3.75) {$A$};

\node (out5) at (1.5,3.75) {$B$};
\node (out6) at (2.5,3.75) {$C$};

\node[style=cprocess] (x1) at (-2.5,2) {$X_i$};
\node[style=cprocess] (x2) at (-.5,2) {$Z_i$};
\node[style=cprocess] (x3) at (1.5,2) {$Z_i$};

\node[style=cprocess,minimum width=1.25cm] (psi1) at (-2,1) {$\Theta_N$};
\node (psi1mid) at (-1.5,2) {};
\node[style=cprocess,minimum width=1.25cm] (psi2) at (0,1) {$\Theta_N$};
\node (psi2mid) at (.5,2) {};
\node[style=cprocess,minimum width=1.25cm] (psi3) at (2,1) {$\Theta_N$};

\node[style=cstate,minimum width = 3.5cm] (L) at (-.75,0) {$L$};

\draw (out1) to[out=-90,in=90] (x1);
\draw (out2) to[out=-90,in=90] (x2);
\draw (out3) to[out=-90,in=90] (x3);

\draw (x1) to (x1|-psi1.north);
\draw (x2) to (x2|-psi2.north);
\draw (x3) to (x3|-psi3.north);

\draw (out4) to[out=-90,in=90] (psi1mid) to (psi1mid|-psi1.north);
\draw (out5) to[out=-90,in=90] (psi2mid) to (psi2mid|-psi2.north);
\draw (out6) to[out=-90,in=90] (out6|-psi3.north);

\draw (psi1) to (psi1|-L.north);
\draw (psi2) to (psi2|-L.north);
\draw (psi3) to (psi3|-L.north);

\end{tikzpicture}
}}
& \mbox{\LARGE $\underset{N^3\sqrt{\epsilon}}{=\joinrel=}$} &
\vcenter{\hbox{
\begin{tikzpicture}

\node (out1) at (-2.5,2.75) {$\mathbf{Q}^{1}$};
\node (out2) at (-1.5,2.75) {$\mathbf{Q}^{2}$};
\node (out3) at (-.5,2.75) {$\mathbf{Q}^{3}$};

\node (out4) at (.5,2.75) {$A$};
\node (out5) at (1.5,2.75) {$B$};
\node (out6) at (2.5,2.75) {$C$};

\node[style=cprocess,minimum width=1.25cm] (psi1) at (-2,1) {$\Theta_N$};
\node[style=cprocess,minimum width=1.25cm] (psi2) at (0,1) {$\Theta_N$};
\node[style=cprocess,minimum width=1.25cm] (psi3) at (2,1) {$\Theta_N$};

\node[style=cstate,minimum width = 3.5cm] (L) at (-.75,0) {$L$};

\draw (out1) to[out=-90,in=90] (psi1.135);
\draw (out2) to[out=-90,in=90] (psi2.135);
\draw (out3) to[out=-90,in=90] (psi3.135);

\draw (psi1) to (psi1|-L.north);
\draw (psi2) to (psi2|-L.north);
\draw (psi3) to (psi3|-L.north);

\draw (out4) to[out=-90,in=90] (psi1.45);
\draw (out5) to[out=-90,in=90] (psi2.45);
\draw (out6) to[out=-90,in=90] (psi3.45);

\end{tikzpicture}}}
\end{array}$
}
\end{equation}
where the last relation also holds for
any permutation of the labels $(X_i, Z_i, Z_i)$ on the left side of the equation.

The commuting reflection operators
$X \otimes Z \otimes Z$,
$Z \otimes X \otimes Z$,
and $Z \otimes Z \otimes X$
on $\mathbb{C}^2 \otimes \mathbb{C}^2
\otimes \mathbb{C}^2$
have a common orthonormal eigenbasis
$G = G_0, G_1, \ldots, G_7$,
in which $G_0$ is the only 
eigenvector that has
eigenvalue $(+1)$ for all three
operators.
We can express $\Theta_A^N\otimes\Theta_B^N\otimes\Theta_c^N\ket{L}$ using this basis as 
\begin{align}
\label{finaldecomp}
\Theta_A^N\otimes\Theta_B^N\otimes\Theta_c^N\ket{L}=\sum_{v_1,\cdots, v_N\in\{G_0, \ldots, G_7 \}}\ket{v_1}\otimes\cdots\otimes\ket{v_N}\otimes\ket{L'_{\mathbf v}}.
\end{align}
where $L'_{\mathbf{v}} \in 
A \otimes B \otimes C \otimes \overline{Q}_{1 \ldots 3N}$.
For every term in the sum on the
right
except the one indexed by
$G_0^{\otimes N}$, there
is an operator
of the form
$X_{A, i} \otimes Z_{B, i} \otimes Z_{C,i}$,
$Z_{A,i} \otimes X_{B,i} \otimes Z_{C,i}$,
or $Z_{A,i} \otimes Z_{B,i} \otimes X_{C,i}$ which negates it.  By equation
(\ref{prop:rigidity XYY}) above,
the total length of all the terms negated by any one particular
gate of this form is
${\cal O} ( N^3 \sqrt{\epsilon })$,
and so the total length of all
terms in  (\ref{finaldecomp})
other than the $G_0^{\otimes N}$
term is ${\cal O} ( N^4 \sqrt{\epsilon})$, as desired.
\end{proof}

We note that our proofs generalize in a straightforward manner to a proof of self-testing for an arbitrary number of copies of a $k$-GHZ state, for any integer $k > 3$. 
The graphical method seems generally
well-suited to proving parallel self-testing for stabilizer
states, including graph states
\cite{mckague2016interactive}.  We leave possible generalizations to future work.

\section{Acknowledgements}

The authors thank Aaron Ostrander and Neil J. Ross for discussions about this project.  This work includes contributions from the National Institute of Standards and Technology, and is not subject to copyright in the United States.

\nocite{*}
\bibliographystyle{eptcs}
\bibliography{rigidity}

\newpage

\begin{appendices}

\section{Categorical Quantum Mechanics}
\label{appendix cqm}
In this appendix we will briefly review some of the standard machinery of categorical quantum mechanics. For a thorough introduction to the topic, see \cite{coecke2017picturing}.

\subsection{Symmetric Monoidal Categories}

The formal context for categorical quantum mechanics is that of \emph{symmetric monoidal categories}. We develop the terminology in stages. 

A \emph{category} is a mathematical structure representing a universe of (possible) processes $F,G,H,\ldots$; each process has a typed input and output, often indicated by writing $F:A\to B$. The fundamental structure in a category is serial composition; whenever the output type of $F:A\to B$ matches the input type of $G:B\to C$, we may form a composite process $G \circ F:A\to C$.

A \emph{monoidal} category generalizes the structure of an ordinary category to allow for multi-partite processes; here the fundamental object of study is a process ${F:A_1\otimes\ldots\otimes A_m \to B_1\otimes\ldots\otimes B_n}$, which is represented as a black box with $m$ labeled inputs and $n$ labeled outputs.
(Either of the numbers $m,n$ may be zero.)  Diagrammatically, we represent these classes as follows:

\begin{center}
\begin{equation}
\begin{tabular}{cccc}
\begin{tikzpicture}
\node[style=cprocess, minimum width=2.5cm] (F) at (0,0) {$F$};

\node (in1) at (-1,-1) {$A_1$};
\draw (in1) to (in1|-F.south);

\node (in2) at (-.33,-1) {$A_2$};
\draw (in2) to (in2|-F.south);

\node (indots) at (.33,-1) {$\ldots$};

\node (in3) at (1,-1) {$A_m$};
\draw (in3) to (in3|-F.south);

\node (out1) at (-1,1) {$B_1$};
\draw (out1) to (out1|-F.north);

\node (out2) at (-.33,1) {$B_2$};
\draw (out2) to (out2|-F.north);

\node (outdots) at (.33,1) {$\ldots$};

\node (out3) at (1,1) {$B_n$};
\draw (out3) to (out3|-F.north);

\end{tikzpicture}
&
\begin{tikzpicture}
\node[style=cstate, minimum width=2cm] (F) at (-.33,0) {$S$};

\node[minimum height={3.5ex}] (in1) at (-1,-1) {};

\node (out1) at (-1,1) {$A_1$};
\draw (out1) to (out1|-F.north);

\node (out2) at (-.33,1) {$A_2$};
\draw (out2) to (out2|-F.north);

\node (outdots) at (.33,1) {$\ldots$};

\node (out3) at (1,1) {$A_m$};
\draw (out3) to (out3|-F.north);
\end{tikzpicture}
&
\begin{tikzpicture}
\node[style=ceffect, minimum width=2cm] (F) at (-.33,0) {$M$};

\node (in1) at (-1,-1) {$B_1$};
\draw (in1) to (in1|-F.south);

\node (in2) at (-.33,-1) {$B_2$};
\draw (in2) to (in2|-F.south);

\node (indots) at (.33,-1) {$\ldots$};

\node (in3) at (1,-1) {$B_n$};
\draw (in3) to (in3|-F.south);

\end{tikzpicture}
&
\begin{tikzpicture}
\node[draw,thick,style=diamond] (F) at (-.33,0) {$K$};

\node[minimum height={3.5ex}] (in1) at (-1,-1) {};
\end{tikzpicture}\\
\end{tabular}
\label{examplediags}
\end{equation}
\end{center}

The categorical structure of serial composition is represented diagrammatically by matching the output wires of one process to the inputs of another, so long as the types match up. So, above, $F$ can be pre-composed with $S$ and post-composed with $M$ to yield a scalar $M \circ F \circ S$ or, in Dirac notation, $\langle M|F|S\rangle$.

Along with multi-partite states and processes, monoidal structure also introduces an operation of parallel composition on processes: given $F:A\to B$ and $G:A'\to B'$, we can produce a parallel process ${F\otimes G: A\otimes A'\to B\otimes B'}$, depicted graphically by side-by-side juxtaposition. More generally, using parallel composition with identity processes (represented by bare wires), we can compose processes in which only some inputs and outputs match. Note that we will often suppress wire labels in complicated diagrams, as the labels are implicitly determined by the boxes they feed.

Finally, a \emph{symmetric} structure on a monoidal category allow for additional flexibility in how wires can be manipulated. A symmetry allows us to permute the ordering of strings, and is represented diagrammatically by crossing wires (called a \emph{twist}). Formally, the twist is axiomatized terms of intuitive diagrammatic equations:

\begin{equation}
\vcenter{\hbox{
\begin{tikzpicture}
\node (in1) at (-.375,-1.5) {$A$};
\node (in2) at (.375,-1.5) {$B$};

\node (out1) at (-.375,1.5) {$A$};
\node (out2) at (.375,1.5) {$B$};

\node (mid1) at (-.375,0) {};
\node (mid2) at (.375,0) {};

\draw (in1) to[in=-90,out=90] (mid2.center);
\draw (mid2.center) to[in=-90,out=90] (out1);

\draw (in2) to[in=-90,out=90] (mid1.center);
\draw (mid1.center) to[in=-90,out=90] (out2);

\node (eq) at (1,0) {=};

\node (in3) at (1.5,-1.5) {$A$};
\node (in4) at (2.25,-1.5) {$B$};

\node (out3) at (1.5,1.5) {$A$};
\node (out4) at (2.25,1.5) {$B$};

\draw (in3) to (out3);
\draw (in4) to (out4);
\end{tikzpicture}
}}
\hspace{2cm}
\vcenter{\hbox{
\begin{tikzpicture}
\node (in1) at (-1,-1.5) {$A$};
\node (in2) at (.25,-1.5) {$B$};

\node (out1) at (-1,1.5) {$A'$};
\node (out2) at (.25,1.5) {$B'$};

\node[style=cprocess] (mid1) at (-1,.75) {$F$};
\node[style=cprocess] (mid2) at (.25,.75) {$G$};

\draw (in1) to[in=-90,out=90] (mid2);
\draw (mid1) to[in=-90,out=90] (out1);

\draw (in2) to[in=-90,out=90] (mid1);
\draw (mid2) to[in=-90,out=90] (out2);

\node (eq) at (1,0) {=};

\node (in3) at (1.75,-1.5) {$A$};
\node (in4) at (3,-1.5) {$B$};

\node (out3) at (1.75,1.5) {$A'$};
\node (out4) at (3,1.5) {$B'$};

\node[style=cprocess] (mid3) at (1.75,-.75) {$G$};
\node[style=cprocess] (mid4) at (3,-.75) {$F$};

\draw (in4) to[in=-90,out=90] (mid4);
\draw (mid4) to[in=-90,out=90] (out3);

\draw (in3) to[in=-90,out=90] (mid3);
\draw (mid3) to[in=-90,out=90] (out4);
\end{tikzpicture}
}}
\label{thetwist}
\end{equation}

\subsection{Quantum states and processes}

To apply the above approach to quantum mechanics, we work within the category of finite-dimensional Hilbert spaces over $\mathbb{C}$.
In this category, which we denote by \textbf{Hilb}, the types are finite-dimensional Hilbert spaces (i.e., vector spaces over $\mathbb{C}$ with semi-linear inner product)
each
equipped with a fixed orthonormal basis, and the processes
are $\mathbb{C}$-linear maps between such vector spaces.  For example, if $A_1, \ldots, A_m, B_1 , \ldots, B_n$ are finite-dimensional
Hilbert spaces, then the diagrams in display box (\ref{examplediags}) above represent, respectively, a linear map $F \colon
A_1 \otimes \ldots \otimes A_m \to B_1 \otimes \ldots \otimes B_n$, a vector $S \in A_1 \otimes \ldots \otimes A_m$,
a linear map $M \colon B_1 \otimes \ldots \otimes B_n \to \mathbb{C}$, and a scalar $K \in \mathbb{C}$.  
Serial composition of diagrams simply represents composition of functions --- for example, the composition of $S, F$ and $M$ is
simply the scalar $M ( F ( S ) ) \in \mathbb{C}$.
It is elementary to show this category is a symmetric monoidal category (with the tensor product as its monoidal operation).

For our purposes, a \textit{state} is a vector in a Hilbert space (i.e., a process with no inputs) whose norm is equal to one. 
A \textit{unitary process} $F \colon A \to B$ is a linear map that satisfies $F F^* = I_B$ and $F^* F = I_A$.  (Also,
a linear map $G \colon A \to B$ that satisfies the single condition $G^*  G = I_A$ is called an \textit{isometry}.)
With these definitions, we will be able to express quantum states and processes as diagrams like the ones in 
(\ref{examplediags}) and (\ref{thetwist}) above.



We note that while symmetric monoidal categories are sufficient to handle the book-keeping needed in our proofs, it is 
only a fragment of the full CQM theory. Further development introduces compact closed structures (trace, transpose, state-process duality), dagger structures (adjoint, conjugate), Frobenius structures (orthonormal basis, classical-quantum interaction) and Hopf structures (complementary bases, ZX-calculus).  For our purposes, we
need only one additional visual definition, which is the \textit{Bell state}.  In this paper we use a gray node with two wires of the same type $R$,
\begin{eqnarray}
\begin{tikzpicture}
\node[style=bell] (bell) at (0,0) {};
\draw (bell) to node [left] {$R$} (-0.5,1);
\draw (bell) to node [right] {$R$} (0.5,1);
\end{tikzpicture}
\end{eqnarray}
to denote the unit vector
\begin{eqnarray}
\left( \frac{1}{\sqrt{\textnormal{dim } R}} \right) \sum_e e \otimes e & \in & R \otimes R,
\end{eqnarray}
where the sum is taken over the standard
basis of $R$.  If $R \cong \mathbb{C}^2$, we
denote this state symbolically by $\Phi^+$.

\section{Controlled Unitaries}
\label{appendix cu}

Our proof makes substantial uses of controlled unitary operations. This appendix collects key facts about controlled operations which will simplify our main proof. 

\begin{definition}[Controlled Unitary]
Let $Q \cong \mathbb{C}^2$ denote a qubit register with a fixed computational basis $\{|0\rangle,|1\rangle\}$, and suppose $U \colon H \to H$ is a unitary operation. The associated \emph{controlled unitary} $\mathbf{C} (U) : Q \otimes H \to Q \otimes H$ is defined by
\begin{eqnarray}
\mathbf{C} ( U ) & = & \left| 0 \right> \left< 0 \right| \otimes I_H + \left| 1 \right> \left< 1 \right| \otimes U.
\end{eqnarray}
\end{definition}

The next lemma describes some (anti-)commutativity properties between controlled unitaries and Pauli operators $X$ and $Z$. The proofs follow directly from the definition.

\begin{lemma}
For any reflection $R:H\to H$ we have the following equations
\begin{equation}
\begin{array}{ccccc}
\vcenter{\hbox{
\begin{tikzpicture}

\node (out1) at (-.5,1.75) {$Q$};
\node (out2) at (.5,1.75) {$H$};

\node (in1) at (-.5,-1) {$Q$};
\node (in2) at (.5,-1) {$H$};

\node[style=cprocess] (x) at (-.5,1) {$Z$};
\node[style=cprocess,minimum width=1.75cm] (CU) at (0,0) {$\mathbf{C}(R)$};

\draw (in1) to (in1|-CU.south);
\draw (out1) to (x) to (x|-CU.north);

\draw (in2) to (in2|-CU.south);
\draw (out2) to (out2|-CU.north);

\end{tikzpicture}
}}
& = &
\vcenter{\hbox{
\begin{tikzpicture}

\node (out1) at (-.5,1.75) {$Q$};
\node (out2) at (.5,1.75) {$H$};

\node (in1) at (-.5,-1) {$Q$};
\node (in2) at (.5,-1) {$H$};

\node[style=cprocess] (x) at (-.5,0) {$Z$};
\node[style=cprocess,minimum width=1.75cm] (CU) at (0,1) {$\mathbf{C}(R)$};

\draw (in1) to (x) to (in1|-CU.south);
\draw (out1) to (x|-CU.north);

\draw (in2) to (in2|-CU.south);
\draw (out2) to (out2|-CU.north);

\end{tikzpicture}
}}
&=&
\vcenter{\hbox{
\begin{tikzpicture}

\node (out1) at (-.5,1.5) {$Q$};
\node (out2) at (.5,1.5) {$H$};

\node (in1) at (-.5,-.5) {$Q$};
\node (in2) at (.5,-.5) {$H$};

\node[style=cprocess,minimum width=1.75cm] (CU) at (0,.5) {$\mathbf{C}(-R)$};

\draw (in1) to (in1|-CU.south);
\draw (out1) to (out1|-CU.north);

\draw (in2) to (in2|-CU.south);
\draw (out2) to (out2|-CU.north);

\end{tikzpicture}
}}\\
\end{array}
\end{equation}

\begin{equation}
\begin{array}{ccccc}
\vcenter{\hbox{
\begin{tikzpicture}

\node (out1) at (-.5,1.75) {$Q$};
\node (out2) at (.5,1.75) {$H$};

\node (in1) at (-.5,-2) {$Q$};
\node (in2) at (.5,-2) {$H$};

\node[style=cprocess] (y1) at (-.5,1) {$X$};
\node[style=cprocess] (y2) at (-.5,-1) {$X$};
\node[style=cprocess,minimum width=1.75cm] (CU) at (0,0) {$\mathbf{C}(R)$};

\draw (in1) to (y2) to (y2|-CU.south);
\draw (out1) to (y1) to (y1|-CU.north);

\draw (in2) to (in2|-CU.south);
\draw (out2) to (out2|-CU.north);

\end{tikzpicture}
}}
& = &
\vcenter{\hbox{
\begin{tikzpicture}

\node (out1) at (-.5,1.25) {$Q$};
\node (out2) at (.5,1.25) {$H$};

\node (in1) at (-.5,-1.5) {$Q$};
\node (in2) at (.5,-1.5) {$H$};

\node[style=cprocess] (U) at (.5,.5) {$R$};
\node[style=cprocess,minimum width=1.75cm] (CU) at (0,-.5) {$\mathbf{C}(R)$};

\draw (in1) to (in1|-CU.south);
\draw (out1) to (out1|-CU.north);

\draw (in2) to (in2|-CU.south);
\draw (out2) to (U) to (U|-CU.north);

\end{tikzpicture}
}}
&=&
\vcenter{\hbox{
\begin{tikzpicture}

\node (out1) at (-.5,1.5) {$Q$};
\node (out2) at (.5,1.5) {$H$};

\node (in1) at (-.5,-1.25) {$Q$};
\node (in2) at (.5,-1.25) {$H$};

\node[style=cprocess] (U) at (.5,-.5) {$R$};
\node[style=cprocess,minimum width=1.75cm] (CU) at (0,.5) {$\mathbf{C}(R)$};

\draw (in1) to (in1|-CU.south);
\draw (out1) to (out1|-CU.north);

\draw (in2) to (U) to (U|-CU.south);
\draw (out2) to (out2|-CU.north);

\end{tikzpicture}
}}\\
\end{array}
\end{equation}
\label{lemma cuxy}
\end{lemma}

\section{Supporting Proofs}
\label{appendix comm proofs}

In this appendix we provide the proofs for propositions from the main text.

\subsection{Proof of Proposition~\ref{prop:app
commutativity}}

We give the proof for $b=0,c=1$; the other cases are analogous. Using inequalities (\ref{keyineqs}) and their variants, we have
\\
\begin{equation}
\begin{array}{ccccc}
\vcenter{\hbox{
\begin{tikzpicture}
\node (out1) at (-1.25,3) {$A$};
\node (out2) at (0,3) {$B$};
\node (out3) at (1.25,3) {$C$};

\node[style=cprocess] (x1) at (-1.25,2) {$Z_{j}'$};

\node[style=cprocess] (y1) at (-1.25,1) {$X_{i}'$};

\node[style=cstate, minimum width=2cm] (L) at (-.5,0) {$L$};

\draw (out1) to (x1) to (y1) to (y1|-L.north);

\draw (out2) to (out2|-L.north);

\draw (out3) to (out3|-L.north);
\end{tikzpicture}
}}
& \mbox{\LARGE $\underset{N\sqrt{\epsilon}}{=\joinrel=}$} &
\vcenter{\hbox{
\begin{tikzpicture}
\node (out1) at (-1.5,2) {$A$};
\node (out2) at (0,2) {$B$};
\node (out3) at (1.5,2) {$C$};


\node[style=cprocess] (y1) at (-1.5,1) {$Z_{j}'$};
\node[style=cprocess] (y2) at (0,1) {$X_{i}'$};
\node[style=cprocess] (y3) at (1.5,1) {$X_{i}'$};

\node[style=cstate, minimum width=2.5cm] (L) at (-.5,0) {$-L$};

\draw (out1) to (y1) to (y1|-L.north);

\draw (out2) to (y2) to (y2|-L.north);

\draw (out3) to (y3) to (y3|-L.north);
\end{tikzpicture}
}}
& \mbox{\LARGE $\underset{N\sqrt{\epsilon}}{=\joinrel=}$} &
\vcenter{\hbox{
\begin{tikzpicture}
\node (out1) at (-1.5,3) {$A$};
\node (out2) at (0,3) {$B$};
\node (out3) at (1.5,3) {$C$};

\node[style=cprocess] (x2) at (0,2) {$X_{i}'$};
\node[style=cprocess] (x3) at (1.5,2) {$X_{i}'$};

\node[style=cprocess] (y2) at (0,1) {$X_{j}'$};
\node[style=cprocess] (y3) at (1.5,1) {$Z_{j}'$};

\node[style=cstate, minimum width=2.5cm] (L) at (-.5,0) {$-L$};

\draw (out1) to (out1|-L.north);

\draw (out2) to (x2) to (y2) to (y2|-L.north);

\draw (out3) to (x3) to (y3) to (y3|-L.north);
\end{tikzpicture}
}}
\end{array}
\end{equation}

\begin{equation*}
\begin{array}{cccccc}
& \mbox{\LARGE $\underset{N\sqrt{\epsilon}}{=\joinrel=}$} &
\vcenter{\hbox{
\begin{tikzpicture}
\node (out1) at (-1.75,2) {$A$};
\node (out2) at (0,2) {$B$};
\node (out3) at (1.5,2) {$C$};


\node[style=cprocess] (y1) at (-1.75,1) {$Z'_{j\mid i \to 0}$};
\node[style=cprocess] (y2) at (0,1) {$X_{i}'$};
\node[style=cprocess] (y3) at (1.5,1) {$X_{i}'$};

\node[style=cstate, minimum width=2.75cm] (L) at (-.75,0) {$-L$};

\draw (out1) to (y1) to (y1|-L.north);

\draw (out2) to (y2) to (y2|-L.north);

\draw (out3) to (y3) to (y3|-L.north);
\end{tikzpicture}
}}
& \mbox{\LARGE $\underset{N\sqrt{\epsilon}}{=\joinrel=}$} &
\vcenter{\hbox{
\begin{tikzpicture}
\node (out1) at (-1.25,3) {$A$};
\node (out2) at (0,3) {$B$};
\node (out3) at (1.25,3) {$C$};

\node[style=cprocess] (x1) at (-1.25,2) {$Z'_{j\mid i \to 0}$};

\node[style=cprocess] (y1) at (-1.25,1) {$X'_{i\mid j \to 1}$};

\node[style=cstate, minimum width=2cm] (L) at (-.5,0) {$L$};

\draw (out1) to (x1) to (y1) to (y1|-L.north);

\draw (out2) to (out2|-L.north);

\draw (out3) to (out3|-L.north);
\end{tikzpicture}
}}
\\
\end{array}
\end{equation*}

A similar sequence shows that 

\begin{equation}
\begin{array}{ccc}
\vcenter{\hbox{
\begin{tikzpicture}
\node (out1) at (-1.25,3) {$A$};
\node (out2) at (0,3) {$B$};
\node (out3) at (1.25,3) {$C$};

\node[style=cprocess] (x1) at (-1.25,2) {$X_{i}'$};

\node[style=cprocess] (y1) at (-1.25,1) {$Z_{j}'$};

\node[style=cstate, minimum width=2cm] (L) at (-.5,0) {$L$};

\draw (out1) to (x1) to (y1) to (y1|-L.north);

\draw (out2) to (out2|-L.north);

\draw (out3) to (out3|-L.north);
\end{tikzpicture}
}}
& \mbox{\LARGE $\underset{N\sqrt{\epsilon}}{=\joinrel=}$} &
\vcenter{\hbox{
\begin{tikzpicture}
\node (out1) at (-1.25,3) {$A$};
\node (out2) at (0,3) {$B$};
\node (out3) at (1.25,3) {$C$};

\node[style=cprocess] (x1) at (-1.25,2) {$X'_{i\mid j \to 1 }$};

\node[style=cprocess] (y1) at (-1.25,1) {$Z'_{j\mid i \to 0}$};

\node[style=cstate, minimum width=2cm] (L) at (-.5,0) {$L$};

\draw (out1) to (x1) to (y1) to (y1|-L.north);

\draw (out2) to (out2|-L.north);

\draw (out3) to (out3|-L.north);
\end{tikzpicture}
}}
\\
\end{array}
\end{equation}

Since $X'_{A, i|j\to 1} =R^A_{i \to 0 \mid j \to 1 }$  and $Z'_{A, j|i\to 0}:=R^A_{j \to 1 \mid i \to 0}$ are assumed to commute, the result follows.

\subsection{Proof of Lemma~\ref{lemma:push}}

Applying the push condition inductively,
we have the following
for some unitary operators $V_1, \ldots, V_k$:
\begin{equation}
\begin{array}{ccccccc}
\vcenter{\hbox{
\begin{tikzpicture}
\node[style=cstate] (z) at
(0,0) {$~~~~Z~~~~~$};
\node[style=cprocess] (uk)
at (-0.3,1) {$U_k$};
\draw (uk) to (uk|-z.north);
\node[style=cprocess] (u1)
at (-0.3,3) {$U_1$};
\draw (uk) to (uk|-z.north);
\node (dots) at (-0.3,2) {$\vdots$};
\draw (u1) to (u1|-dots.north);
\draw (dots) to (dots|-uk.north);
\node[style=cprocess] (v) at (-0.3,4) {$V$};
\draw (v) to (v|-u1.north);
\node (r) at (-0.3,5) {$R$};
\node (s) at (0.5,5) {$S$};
\draw (s) to (s|-z.north);
\draw (r) to (r|-v.north);
\end{tikzpicture}}}
& \vcenter{\hbox{{\LARGE $\underset{k\epsilon}{=\joinrel=}$}}} & 
\vcenter{\hbox{
\begin{tikzpicture}
\node[style=cstate] (z) at
(0,0) {$~~~~Z~~~~~$};
\node[style=cprocess] (uk)
at (0.5,3) {$V_k$};
\draw (uk) to (uk|-dots.north);
\node[style=cprocess] (u1)
at (0.5,1) {$V_1$};
\node (dots) at (0.5,2) {$\vdots$};
\draw (dots) to (dots|-u1.north);
\node[style=cprocess] (v) at (-0.3,2) {$V$};
\draw (v) to (v|-u1.north);
\node (r) at (-0.3,5) {$R$};
\node (s) at (0.5,5) {$S$};
\draw (s) to (s|-uk.north);
\draw (r) to (r|-v.north);
\draw (v) to (v|-z.north);
\draw (u1) to (u1|-z.north);
\end{tikzpicture}}}
& \vcenter{\hbox{{\LARGE $\underset{\delta}{=\joinrel=}$}}} & 
\vcenter{\hbox{
\begin{tikzpicture}
\node[style=cstate] (z) at
(0,0) {$~~~~Z~~~~~$};
\node[style=cprocess] (uk)
at (0.5,3) {$V_k$};
\draw (uk) to (uk|-dots.north);
\node[style=cprocess] (u1)
at (0.5,1) {$V_1$};
\node (dots) at (0.5,2) {$\vdots$};
\draw (dots) to (dots|-u1.north);
\node[style=cprocess] (w) at (-0.3,2) {$W$};
\draw (w) to (w|-u1.north);
\node (r) at (-0.3,5) {$R$};
\node (s) at (0.5,5) {$S$};
\draw (s) to (s|-uk.north);
\draw (r) to (r|-w.north);
\draw (w) to (w|-z.north);
\draw (u1) to (u1|-z.north);
\end{tikzpicture}}}
& \vcenter{\hbox{{\LARGE $\underset{k\epsilon}{=\joinrel=}$}}} & 
\vcenter{\hbox{
\begin{tikzpicture}
\node[style=cstate] (z) at
(0,0) {$~~~~Z~~~~~$};
\node[style=cprocess] (uk)
at (-0.3,1) {$U_k$};
\draw (uk) to (uk|-z.north);
\node[style=cprocess] (u1)
at (-0.3,3) {$U_1$};
\draw (uk) to (uk|-z.north);
\node (dots) at (-0.3,2) {$\vdots$};
\draw (u1) to (u1|-dots.north);
\draw (dots) to (dots|-uk.north);
\node[style=cprocess] (w) at (-0.3,4) {$W$};
\draw (w) to (w|-u1.north);
\node (r) at (-0.3,5) {$R$};
\node (s) at (0.5,5) {$S$};
\draw (s) to (s|-z.north);
\draw (r) to (r|-w.north);
\end{tikzpicture}}}

\end{array}
\end{equation}

\subsection{Proof of Proposition~\ref{prop:correct pauli}}

\label{app:correct pauli}

We begin with two observations.  First, the 
approximate anti-commutativity in Proposition~\ref{prop:antimulti} implies the following approximate anti-commutativity property for the controlled-$X_k'$ gate (by superposition):

\begin{equation}
\begin{array}{ccc}
\vcenter{\hbox{
\begin{tikzpicture}
\node[style=cprocess] (l) at (1,-3.5) {$~~~~~~~L~~~~~~~$};

\node[style=bell] (bell) at (-1.2,-3.5) {};

\node[style=cprocess,minimum width = 1.5cm] (CX) at (0,-1.5) {$\mathbf{C}(X_k')$};

\draw (l.150) to[out=90,in=-90] node[left] {$A$} (CX);

\draw (bell) to[in=-90, out=45] (CX.230);
\draw (bell) to node[left] {$\overline{Q}_k$} (-1.2, 1);

\node[style=cprocess] (z) at (0,0) {$Z'_k$};

\draw (z) to (z|-CX.north);

\draw (z) to (0,1);

\draw (l.90) to[out=90, in=-90] (2,1);
\draw (l.50) to[out=90, in=-90] (3,1);

\draw (CX.130) to[out=90,in=-90] (-0.6,1);

\end{tikzpicture}
}} 
&
\mbox{\LARGE $\underset{N\sqrt{\epsilon}}{=\joinrel=}$}
&
\vcenter{\hbox{
\begin{tikzpicture}
\node[style=cprocess] (l) at (1,-3.5) {$~~~~~~~L~~~~~~~$};

\node[style=bell] (bell) at (-1.2,-3.5) {};

\node[style=cprocess] (z) at (0,-1.5) {$Z'_k$};

\node[style=cprocess,minimum width = 1.5cm] (CX) at (0,0) {$\mathbf{C}(-X_k')$};
\draw (z) to (z|-CX.south);

\draw (l.150) to[out=90,in=-90] node[left] {$A$} (z);

\draw (bell) to[in=-90, out=75] (CX.230);
\draw (bell) to node[left] {$\overline{Q}_k$} (-1.2, 1);

\draw (CX) to (0,1);
\draw (CX.130) to[out=90,in=-90] (-0.6,1);

\draw (l.90) to[out=90, in=-90] (2,1);
\draw (l.50) to[out=90, in=-90] (3,1);

\end{tikzpicture}
}} 
\end{array}
\label{controlledanti}
\end{equation}

Secondly, the controlled operator 
$\mathbf{C} (X'_{A, k} )$ on $Q_k \otimes A$
can be approximately pushed through the state $\Phi^+ \otimes L$ like so:
\begin{equation}
\begin{array}{ccc}
\vcenter{\hbox{
\begin{tikzpicture}
\node[style=cprocess] (l) at (1,-3.5) {$~~~~~~~L~~~~~~~$};

\node[style=bell] (bell) at (0.8,-1.5) {};

\node[style=cprocess,minimum width = 1.5cm] (CX) at (0,0) {$\mathbf{C}(X_k')$};

\draw (l.150) to[out=90,in=-90]  (CX);

\draw (bell) to[in=-90, out=135] (CX.310);

\node (oq) at (1.2,1.2) {$\overline{Q}_k$};
\draw (bell) to[out=45, in=-90]  (oq);

\node (a) at (-0.5,1.2) {$A$};
\draw (CX) to (a);

\node (q) at (0.3,1.2) {$Q_k$};
\draw (CX.50) to (q);

\node (b) at (2,1.2) {$B$};
\draw (l.90) to[out=90, in=-90] (b);

\node (c) at (3,1.2) {$C$};
\draw (l.50) to[out=90, in=-90] (c);

\end{tikzpicture}
}} 
&
\mbox{\LARGE $\underset{N\sqrt{\epsilon}}{=\joinrel=}$}
&
\vcenter{\hbox{
\begin{tikzpicture}
\node[style=cprocess] (l) at (1,-3.5) {$~~~~~~~L~~~~~~~$};

\node[style=bell] (bell) at (0.8,-1.5) {};

\node[style=cprocess,minimum width = 1.5cm] (CX) at (2,0) {$\mathbf{C}(Z_{B,k}' \otimes Z_{C,k}' )$};


\draw (bell) to[in=-90, out=45] (CX.210);

\node (q) at (0.3,1.2) {$Q_k$};
\node (oq) at (1.2,1.2) {$\overline{Q}_k$};
\draw (bell) to[out=135, in=-90]  (q);

\node (a) at (-0.5,1.2) {$A$};
\draw (l.130) to (a);

\draw (CX.150) to (oq);

\node (b) at (2,1.2) {$B$};
\draw (CX.90) to[out=90, in=-90] (b);
\draw (l.90) to[out=90, in=-90] (CX.270);

\node (c) at (3,1.2) {$C$};
\draw (CX.50) to (c);
\draw (l.50) to[out=90, in=-90] (CX.310);

\end{tikzpicture}
}} 

\end{array}
\end{equation}
And similarly for $Z'_{A,k}$.  The Hadamard
operator $[H \otimes I_A]$ on $Q_k \otimes A$
can be exactly pushed through $\Phi^+ \otimes L$ 
(by merely applying $H$ to $\overline{Q}_k$).  This fact
allows free application of the Push Lemma (Lemma~\ref{lemma:push}).  We have the following,
in which we exploit approximate anti-commutativity,
the Lemma~\ref{lemma:push}, and the rules
for controlled unitaries from Appendix~\ref{appendix cu}.
\begin{equation}
\begin{array}{ccc}
\vcenter{\hbox{
\begin{tikzpicture}

\node[style=cprocess] (l) at (1,-1.5) {$~~~~~~L~~~~~~$};

\node (oq) at (-1,2) {$\overline{Q}_k$};
\node (q) at (0,2) {$Q_k$};
\node (a) at (1,2) {$A$};

\node[style=cprocess] (x) at (0,1) {X};

\node[style=cprocess,minimum width = 1.5cm] (psi) at (0,0) {$\Psi_{k}$};

\draw (l.130) to[out=90, in=-90] (psi.south);

\draw (psi.130) to (oq);
\draw (psi.90) to (x);
\draw (x) to (q);
\draw (psi.50) to (a);

\node (b) at (1.5,0) {$B$};
\node (c) at (2.5,0) {$C$};
\draw (l.90) to[in=-90,out=90]  (b);
\draw (l.50) to[in=-90,out=90]  (c);
\end{tikzpicture}
}}
& \mbox{\LARGE $=\joinrel=$}
&
\vcenter{\hbox{
\begin{tikzpicture}
\node[style=cstate] (l) at (1,-3) {$~~~~~L~~~~~$};

\node[style=bell] (bell) at (-1.2,-2.5) {};

\node (q) at (-.5,4.5) {$Q_k$};
\node (a) at (.5,4.5) {$A$};
\node (oq) at (-1.2,4.5) {$\overline{Q}_k$};

\node[style=cprocess] (x) at (-.5,3.5) {$X$};

\node[style=cprocess,minimum width = 1.5cm] (cz) at (0,0.5) {$\mathbf{C}(Z_k')$};
\node[style=cprocess] (h2) at (-.5,1.5) {$H$};
\node[style=cprocess,minimum width = 1.5cm] (cx2) at (0,2.5) {$\mathbf{C}(X_k')$};

\node[style=cprocess] (h1) at (-.5,-0.5) {$H$};

\node[style=cprocess,minimum width = 1.5cm] (cx1) at (0,-1.5) {$\mathbf{C}(X_k')$};

\draw (l) to (a|-cx1.south);

\draw (a|-cx1.north) to (a|-cz.south);

\draw (a|-cz.north) to (a|-cx2.south);

\draw (bell) to[in=-90, out=0] (cx1.220);
\draw (h1) to (h1|-cx1.north);
\draw (h1) to (h1|-cz.south);

\draw (h2) to (h2|-cz.north);
\draw (h2) to (h2|-cx2.south);

\draw (bell) to (oq);

\draw (h1|-cx2.north) to (h1|-x.south);
\draw (h1|-x.north) to (h1|-q.south);
\draw (a|-cx2.north) to (a|-a.south);

\node (b) at (1.5,0) {$B$};
\node (c) at (2.5,0) {$C$};
\draw (l.90) to[in=-90,out=90]  (b);
\draw (l.50) to[in=-90,out=90]  (c);

\end{tikzpicture}
}} 
\end{array}
\end{equation}

\begin{equation}
\begin{array}{cccc}
\mbox{\LARGE $=\joinrel=$}
&
\vcenter{\hbox{
\begin{tikzpicture}
\node[style=cstate] (l) at (1,-3) {$~~~~~L~~~~~$};

\node[style=bell] (bell) at (-1.2,-2.5) {};

\node (q) at (-.5,4.5) {$Q_k$};
\node (a) at (.5,4.5) {$A$};
\node (oq) at (-1.2,4.5) {$\overline{Q}_k$};

\node[style=cprocess] (x) at (-.5,2.5) {$X$};

\node[style=cprocess,minimum width = 1.5cm] (cz) at (0,0.5) {$\mathbf{C}(Z_k')$};
\node[style=cprocess] (h2) at (-.5,1.5) {$H$};
\node[style=cprocess,minimum width = 1.5cm] (cx2) at (0,3.5) {$\mathbf{C}(X_k')$};

\node[style=cprocess] (h1) at (-.5,-0.5) {$H$};

\node[style=cprocess,minimum width = 1.5cm] (cx1) at (0,-1.5) {$\mathbf{C}(X_k')$};

\node[style=cprocess] (xp) at (0.5,2.5) {$X'_k$};

\draw (l) to (a|-cx1.south);

\draw (a|-cx1.north) to (a|-cz.south);

\draw (a|-cz.north) to (a|-xp.south);
\draw (a|-xp.north) to (a|-cx2.south);

\draw (bell) to[in=-90, out=0] (cx1.220);
\draw (h1) to (h1|-cx1.north);
\draw (h1) to (h1|-cz.south);

\draw (h2) to (h2|-cz.north);
\draw (h2) to (h2|-x.south);

\draw (bell) to (oq);

\draw (h1|-cx2.south) to (h1|-x.north);
\draw (h1|-cx2.north) to (h1|-q.south);
\draw (a|-cx2.north) to (a|-a.south);

\node (b) at (1.5,0) {$B$};
\node (c) at (2.5,0) {$C$};
\draw (l.90) to[in=-90,out=90]  (b);
\draw (l.50) to[in=-90,out=90]  (c);

\end{tikzpicture}
}} 
& \mbox{\LARGE $\underset{N\sqrt{\epsilon}}{=\joinrel=}$}
&
\vcenter{\hbox{
\begin{tikzpicture}
\node[style=cstate] (l) at (1,-3) {$~~~~~L~~~~~$};

\node[style=bell] (bell) at (-1.2,-2.5) {};

\node (q) at (-.5,4.5) {$Q_k$};
\node (a) at (.5,4.5) {$A$};
\node (oq) at (-1.2,4.5) {$\overline{Q}_k$};

\node[style=cprocess] (z) at (-.5,1.5) {$Z$};

\node[style=cprocess,minimum width = 1.5cm] (cz) at (0,0.5) {$\mathbf{C}(-Z_k')$};
\node[style=cprocess] (h2) at (-.5,2.5) {$H$};
\node[style=cprocess,minimum width = 1.5cm] (cx2) at (0,3.5) {$\mathbf{C}(X_k')$};

\node[style=cprocess] (h1) at (-.5,-0.5) {$H$};

\node[style=cprocess,minimum width = 1.5cm] (cx1) at (0,-1.5) {$\mathbf{C}(X_k')$};

\node[style=cprocess] (xp) at (0.5,-0.5) {$X'_k$};

\draw (l) to (a|-cx1.south);

\draw (xp|-xp.north) to (xp|-cz.south);
\draw (xp|-xp.south) to (xp|-cx1.north);

\draw (a|-cz.north) to (a|-cx2.south);

\draw (bell) to[in=-90, out=0] (cx1.220);
\draw (h1) to (h1|-cx1.north);
\draw (h1) to (h1|-cz.south);

\draw (h2) to (h2|-cx2.south);
\draw (h2) to (h2|-z.north);

\draw (z|-cz.north) to (z|-z.south);

\draw (bell) to (oq);

\draw (h1|-cx2.north) to (h1|-q.south);
\draw (a|-cx2.north) to (a|-a.south);

\node (b) at (1.5,0) {$B$};
\node (c) at (2.5,0) {$C$};
\draw (l.90) to[in=-90,out=90]  (b);
\draw (l.50) to[in=-90,out=90]  (c);

\end{tikzpicture}
}} 

\end{array}
\end{equation}

\begin{equation}
\begin{array}{cccc}
\mbox{\LARGE $=\joinrel=$}
&
\vcenter{\hbox{
\begin{tikzpicture}
\node[style=cstate] (l) at (1,-3) {$~~~~~L~~~~~$};

\node[style=bell] (bell) at (-1.2,-2.5) {};

\node (q) at (-.5,4.5) {$Q_k$};
\node (a) at (.5,4.5) {$A$};
\node (oq) at (-1.2,4.5) {$\overline{Q}_k$};

\node[style=cprocess,minimum width = 1.5cm] (cz) at (0,1.5) {$\mathbf{C}(Z_k')$};
\node[style=cprocess] (h2) at (-.5,2.5) {$H$};
\node[style=cprocess,minimum width = 1.5cm] (cx2) at (0,3.5) {$\mathbf{C}(X_k')$};

\node[style=cprocess] (h1) at (-.5,0.5) {$H$};

\node[style=cprocess,minimum width = 1.5cm] (cx1) at (0,-0.5) {$\mathbf{C}(X_k')$};

\node[style=cprocess] (xp) at (0.5,-1.5) {$X'_k$};

\draw (l) to (a|-xp.south);

\draw (xp|-xp.north) to (xp|-cx1.south);

\draw (a|-cz.north) to (a|-cx2.south);

\draw (bell) to[in=-90, out=0] (cx1.220);
\draw (h1) to (h1|-cx1.north);
\draw (h1) to (h1|-cz.south);

\draw (h2) to (h2|-cx2.south);
\draw (h2) to (h2|-cz.north);

\draw (bell) to (oq);

\draw (h1|-cx2.north) to (h1|-q.south);
\draw (a|-cx2.north) to (a|-a.south);

\draw (xp|-cx1.north) to (xp|-cz.south);

\node (b) at (1.5,0) {$B$};
\node (c) at (2.5,0) {$C$};
\draw (l.90) to[in=-90,out=90]  (b);
\draw (l.50) to[in=-90,out=90]  (c);

\end{tikzpicture}
}} 

& \mbox{\LARGE $=\joinrel=$}
&

\vcenter{\hbox{
\begin{tikzpicture}

\node[style=cprocess] (l) at (1,-1.5) {$~~~~~~L~~~~~~$};

\node (oq) at (-1,2) {$\overline{Q}_k$};
\node (q) at (0,2) {$Q_k$};
\node (a) at (1,2) {$A$};
\node (b) at (1.5,0) {$B$};
\node (c) at (2.5,0) {$C$};

\node[style=cprocess] (xp) at (0.3,-0.2) {$X'_k$};

\node[style=cprocess,minimum width = 1.5cm] (psi) at (0,0.8) {$\Psi_{k}$};

\draw (l.130) to (xp);
\draw (xp) to (psi.270);

\draw (psi.130) to (oq);
\draw (psi.50) to (a);
\draw (psi.90) to (q);

\draw (l.90) to[in=-90,out=90]  (b);
\draw (l.50) to[in=-90,out=90]  (c);
\end{tikzpicture}
}}
\end{array}
\end{equation}

Similarly,

\begin{equation}
\begin{array}{ccc}
\vcenter{\hbox{
\begin{tikzpicture}

\node[style=cprocess] (l) at (1,-1.5) {$~~~~~~L~~~~~~$};

\node (oq) at (-1,2) {$\overline{Q}_k$};
\node (q) at (0,2) {$Q_k$};
\node (a) at (1,2) {$A$};

\node[style=cprocess] (z) at (0,1) {Z};

\node[style=cprocess,minimum width = 1.5cm] (psi) at (0,0) {$\Psi_{k}$};

\draw (l.130) to[out=90, in=-90] (psi.south);

\draw (psi.130) to (oq);
\draw (psi.90) to (z);
\draw (z) to (q);
\draw (psi.50) to (a);

\node (b) at (1.5,0) {$B$};
\node (c) at (2.5,0) {$C$};
\draw (l.90) to[in=-90,out=90]  (b);
\draw (l.50) to[in=-90,out=90]  (c);
\end{tikzpicture}
}}
& \mbox{\LARGE $=\joinrel=$}
&
\vcenter{\hbox{
\begin{tikzpicture}
\node[style=cstate] (l) at (1,-3) {$~~~~~L~~~~~$};

\node[style=bell] (bell) at (-1.2,-2.5) {};

\node (q) at (-.5,4.5) {$Q_k$};
\node (a) at (.5,4.5) {$A$};
\node (oq) at (-1.2,4.5) {$\overline{Q}_k$};

\node[style=cprocess] (z) at (-.5,3.5) {$Z$};

\node[style=cprocess,minimum width = 1.5cm] (cz) at (0,0.5) {$\mathbf{C}(Z_k')$};
\node[style=cprocess] (h2) at (-.5,1.5) {$H$};
\node[style=cprocess,minimum width = 1.5cm] (cx2) at (0,2.5) {$\mathbf{C}(X_k')$};

\node[style=cprocess] (h1) at (-.5,-0.5) {$H$};

\node[style=cprocess,minimum width = 1.5cm] (cx1) at (0,-1.5) {$\mathbf{C}(X_k')$};

\draw (l) to (a|-cx1.south);

\draw (a|-cx1.north) to (a|-cz.south);

\draw (a|-cz.north) to (a|-cx2.south);

\draw (bell) to[in=-90, out=0] (cx1.220);
\draw (h1) to (h1|-cx1.north);
\draw (h1) to (h1|-cz.south);

\draw (h2) to (h2|-cz.north);
\draw (h2) to (h2|-cx2.south);

\draw (bell) to (oq);

\draw (h1|-cx2.north) to (h1|-z.south);
\draw (h1|-z.north) to (h1|-q.south);
\draw (a|-cx2.north) to (a|-a.south);

\node (b) at (1.5,0) {$B$};
\node (c) at (2.5,0) {$C$};
\draw (l.90) to[in=-90,out=90]  (b);
\draw (l.50) to[in=-90,out=90]  (c);

\end{tikzpicture}
}} 
\end{array}
\end{equation}

\begin{equation}
\begin{array}{cccc}
\mbox{\LARGE $=\joinrel=$}
&
\vcenter{\hbox{
\begin{tikzpicture}
\node[style=cstate] (l) at (1,-3) {$~~~~~L~~~~~$};

\node[style=bell] (bell) at (-1.2,-2.5) {};

\node (q) at (-.5,4.5) {$Q_k$};
\node (a) at (.5,4.5) {$A$};
\node (oq) at (-1.2,4.5) {$\overline{Q}_k$};

\node[style=cprocess] (z) at (-.5,2.5) {$Z$};

\node[style=cprocess,minimum width = 1.5cm] (cz) at (0,0.5) {$\mathbf{C}(Z_k')$};
\node[style=cprocess] (h2) at (-.5,1.5) {$H$};
\node[style=cprocess,minimum width = 1.5cm] (cx2) at (0,3.5) {$\mathbf{C}(X_k')$};

\node[style=cprocess] (h1) at (-.5,-0.5) {$H$};

\node[style=cprocess,minimum width = 1.5cm] (cx1) at (0,-1.5) {$\mathbf{C}(X_k')$};

\draw (l) to (a|-cx1.south);

\draw (a|-cx1.north) to (a|-cz.south);

\draw (a|-cz.north) to (a|-cx2.south);

\draw (bell) to[in=-90, out=0] (cx1.220);
\draw (h1) to (h1|-cx1.north);
\draw (h1) to (h1|-cz.south);

\draw (h2) to (h2|-cz.north);
\draw (h2) to (h2|-z.south);

\draw (bell) to (oq);

\draw (h1|-cx2.south) to (h1|-z.north);
\draw (h1|-cx2.north) to (h1|-q.south);
\draw (a|-cx2.north) to (a|-a.south);

\node (b) at (1.5,0) {$B$};
\node (c) at (2.5,0) {$C$};
\draw (l.90) to[in=-90,out=90]  (b);
\draw (l.50) to[in=-90,out=90]  (c);

\end{tikzpicture}
}} 

& 

\mbox{\LARGE $=\joinrel=$}
&
\vcenter{\hbox{
\begin{tikzpicture}
\node[style=cstate] (l) at (1,-3) {$~~~~~L~~~~~$};

\node[style=bell] (bell) at (-1.2,-2.5) {};

\node (q) at (-.5,4.5) {$Q_k$};
\node (a) at (.5,4.5) {$A$};
\node (oq) at (-1.2,4.5) {$\overline{Q}_k$};

\node[style=cprocess] (x) at (-.5,1.5) {$X$};

\node[style=cprocess,minimum width = 1.5cm] (cz) at (0,0.5) {$\mathbf{C}(Z_k')$};
\node[style=cprocess] (h2) at (-.5,2.5) {$H$};
\node[style=cprocess,minimum width = 1.5cm] (cx2) at (0,3.5) {$\mathbf{C}(X_k')$};

\node[style=cprocess] (h1) at (-.5,-0.5) {$H$};

\node[style=cprocess,minimum width = 1.5cm] (cx1) at (0,-1.5) {$\mathbf{C}(X_k')$};

\draw (l) to (a|-cx1.south);

\draw (a|-cx1.north) to (a|-cz.south);

\draw (a|-cz.north) to (a|-cx2.south);

\draw (bell) to[in=-90, out=0] (cx1.220);
\draw (h1) to (h1|-cx1.north);
\draw (h1) to (h1|-cz.south);

\draw (x) to (h2|-cz.north);
\draw (h2) to (h2|-x.north);

\draw (bell) to (oq);

\draw (h1|-cx2.south) to (h1|-z.north);
\draw (h1|-cx2.north) to (h1|-q.south);
\draw (a|-cx2.north) to (a|-a.south);

\node (b) at (1.5,0) {$B$};
\node (c) at (2.5,0) {$C$};
\draw (l.90) to[in=-90,out=90]  (b);
\draw (l.50) to[in=-90,out=90]  (c);

\end{tikzpicture}
}} 

\end{array}
\end{equation}

\begin{equation}
\begin{array}{cccc}
\mbox{\LARGE $=\joinrel=$}
&
\vcenter{\hbox{
\begin{tikzpicture}
\node[style=cstate] (l) at (1,-3) {$~~~~~L~~~~~$};

\node[style=bell] (bell) at (-1.2,-2.5) {};

\node (q) at (-.5,4.5) {$Q_k$};
\node (a) at (.5,4.5) {$A$};
\node (oq) at (-1.2,4.5) {$\overline{Q}_k$};

\node[style=cprocess] (x) at (-.5,0.5) {$X$};

\node[style=cprocess,minimum width = 1.5cm] (cz) at (0,1.5) {$\mathbf{C}(Z_k')$};
\node[style=cprocess] (h2) at (-.5,2.5) {$H$};
\node[style=cprocess,minimum width = 1.5cm] (cx2) at (0,3.5) {$\mathbf{C}(X_k')$};

\node[style=cprocess] (h1) at (-.5,-0.5) {$H$};

\node[style=cprocess,minimum width = 1.5cm] (cx1) at (0,-1.5) {$\mathbf{C}(X_k')$};

\node[style=cprocess] (zp) at (.5,0.5) {$Z_k'$};

\draw (l) to (a|-cx1.south);

\draw (a|-cx1.north) to (a|-zp.south);
\draw (a|-cz.south) to (a|-zp.north);

\draw (a|-cz.north) to (a|-cx2.south);

\draw (bell) to[in=-90, out=0] (cx1.220);
\draw (h1) to (h1|-cx1.north);
\draw (h1) to (h1|-x.south);

\draw (x) to (h2|-cz.south);
\draw (h2) to (h2|-cz.north);

\draw (bell) to (oq);

\draw (h1|-cx2.south) to (h1|-z.north);
\draw (h1|-cx2.north) to (h1|-q.south);
\draw (a|-cx2.north) to (a|-a.south);

\node (b) at (1.5,0) {$B$};
\node (c) at (2.5,0) {$C$};
\draw (l.90) to[in=-90,out=90]  (b);
\draw (l.50) to[in=-90,out=90]  (c);

\end{tikzpicture}
}} 
&

\mbox{\LARGE $=\joinrel=$}
&
\vcenter{\hbox{
\begin{tikzpicture}
\node[style=cstate] (l) at (1,-3) {$~~~~~L~~~~~$};

\node[style=bell] (bell) at (-1.2,-2.5) {};

\node (q) at (-.5,4.5) {$Q_k$};
\node (a) at (.5,4.5) {$A$};
\node (oq) at (-1.2,4.5) {$\overline{Q}_k$};

\node[style=cprocess] (z) at (-.5,-0.5) {$Z$};

\node[style=cprocess,minimum width = 1.5cm] (cz) at (0,1.5) {$\mathbf{C}(Z_k')$};
\node[style=cprocess] (h2) at (-.5,2.5) {$H$};
\node[style=cprocess,minimum width = 1.5cm] (cx2) at (0,3.5) {$\mathbf{C}(X_k')$};

\node[style=cprocess] (h1) at (-.5,0.5) {$H$};

\node[style=cprocess,minimum width = 1.5cm] (cx1) at (0,-1.5) {$\mathbf{C}(X_k')$};

\node[style=cprocess] (zp) at (.5,0.5) {$Z_k'$};

\draw (l) to (a|-cx1.south);

\draw (a|-cx1.north) to (a|-zp.south);
\draw (a|-cz.south) to (a|-zp.north);

\draw (a|-cz.north) to (a|-cx2.south);

\draw (bell) to[in=-90, out=0] (cx1.220);
\draw (h1) to (h1|-cz.south);
\draw (h1) to (h1|-z.north);

\draw (h2) to (h2|-cx2.south);
\draw (h2) to (h2|-cz.north);

\draw (bell) to (oq);

\draw (h1|-cx1.north) to (h1|-z.south);
\draw (h1|-cx2.north) to (h1|-q.south);
\draw (a|-cx2.north) to (a|-a.south);

\node (b) at (1.5,0) {$B$};
\node (c) at (2.5,0) {$C$};
\draw (l.90) to[in=-90,out=90]  (b);
\draw (l.50) to[in=-90,out=90]  (c);

\end{tikzpicture}
}} 

\end{array}
\end{equation}

\begin{equation}
\begin{array}{cccc}
\mbox{\LARGE $=\joinrel=$}
&
\vcenter{\hbox{
\begin{tikzpicture}
\node[style=cstate] (l) at (1,-3) {$~~~~~L~~~~~$};

\node[style=bell] (bell) at (-1.2,-2.5) {};

\node (q) at (-.5,4.5) {$Q_k$};
\node (a) at (.5,4.5) {$A$};
\node (oq) at (-1.2,4.5) {$\overline{Q}_k$};

\node[style=cprocess,minimum width = 1.5cm] (cz) at (0,1.5) {$\mathbf{C}(Z_k')$};
\node[style=cprocess] (h2) at (-.5,2.5) {$H$};
\node[style=cprocess,minimum width = 1.5cm] (cx2) at (0,3.5) {$\mathbf{C}(X_k')$};

\node[style=cprocess] (h1) at (-.5,0.5) {$H$};

\node[style=cprocess,minimum width = 1.5cm] (cx1) at (0,-0.5) {$\mathbf{C}(-X_k')$};

\node[style=cprocess] (zp) at (.5,0.5) {$Z_k'$};

\draw (l) to (a|-cx1.south);

\draw (a|-cx1.north) to (a|-zp.south);
\draw (a|-cz.south) to (a|-zp.north);

\draw (a|-cz.north) to (a|-cx2.south);

\draw (bell) to[in=-90, out=0] (cx1.220);
\draw (h1) to (h1|-cz.south);
\draw (h1) to (h1|-cx1.north);

\draw (h2) to (h2|-cx2.south);
\draw (h2) to (h2|-cz.north);

\draw (bell) to (oq);

\draw (h1|-cx2.north) to (h1|-q.south);
\draw (a|-cx2.north) to (a|-a.south);

\node (b) at (1.5,0) {$B$};
\node (c) at (2.5,0) {$C$};
\draw (l.90) to[in=-90,out=90]  (b);
\draw (l.50) to[in=-90,out=90]  (c);

\end{tikzpicture}
}} 

&

\mbox{\LARGE $\underset{N \sqrt{\epsilon}}{=\joinrel=}$}
&
\vcenter{\hbox{
\begin{tikzpicture}
\node[style=cstate] (l) at (1,-3) {$~~~~~L~~~~~$};

\node[style=bell] (bell) at (-1.2,-2.5) {};

\node (q) at (-.5,4.5) {$Q_k$};
\node (a) at (.5,4.5) {$A$};
\node (oq) at (-1.2,4.5) {$\overline{Q}_k$};

\node[style=cprocess,minimum width = 1.5cm] (cz) at (0,1.5) {$\mathbf{C}(Z_k')$};
\node[style=cprocess] (h2) at (-.5,2.5) {$H$};
\node[style=cprocess,minimum width = 1.5cm] (cx2) at (0,3.5) {$\mathbf{C}(X_k')$};

\node[style=cprocess] (h1) at (-.5,0.5) {$H$};

\node[style=cprocess,minimum width = 1.5cm] (cx1) at (0,-0.5) {$\mathbf{C}(X_k')$};

\node[style=cprocess] (zp) at (.5,-1.5) {$Z_k'$};

\draw (l) to (a|-zp.south);

\draw (a|-cx1.south) to (a|-zp.north);
\draw (a|-cz.south) to (a|-cx1.north);

\draw (a|-cz.north) to (a|-cx2.south);

\draw (bell) to[in=-90, out=0] (cx1.220);
\draw (h1) to (h1|-cz.south);
\draw (h1) to (h1|-cx1.north);

\draw (h2) to (h2|-cx2.south);
\draw (h2) to (h2|-cz.north);

\draw (bell) to (oq);

\draw (h1|-cx2.north) to (h1|-q.south);
\draw (a|-cx2.north) to (a|-a.south);

\node (b) at (1.5,0) {$B$};
\node (c) at (2.5,0) {$C$};
\draw (l.90) to[in=-90,out=90]  (b);
\draw (l.50) to[in=-90,out=90]  (c);

\end{tikzpicture}
}}

\end{array}
\end{equation}

\begin{equation}
\begin{array}{cc}
\mbox{\LARGE $=\joinrel=$}
&
\vcenter{\hbox{
\begin{tikzpicture}

\node[style=cprocess] (l) at (1,-1.5) {$~~~~~~L~~~~~~$};

\node (oq) at (-1,2) {$\overline{Q}_k$};
\node (q) at (0,2) {$Q_k$};
\node (a) at (1,2) {$A$};
\node (b) at (1.5,0) {$B$};
\node (c) at (2.5,0) {$C$};

\node[style=cprocess] (zp) at (0.3,-0.2) {$Z'_k$};

\node[style=cprocess,minimum width = 1.5cm] (psi) at (0,0.8) {$\Psi_{k}$};

\draw (l.130) to (zp);
\draw (zp) to (psi.270);

\draw (psi.130) to (oq);
\draw (psi.50) to (a);
\draw (psi.90) to (q);

\draw (l.90) to[in=-90,out=90]  (b);
\draw (l.50) to[in=-90,out=90]  (c);
\end{tikzpicture}
}}
\end{array}
\end{equation}
as desired.  This completes the proof.

\subsection{Proof of Proposition~\ref{prop:multi pauli}}

We begin with the following lemma.  It is similar to
the Push Lemma (Lemma~\ref{lemma:push}) but it specifically
addresses commutativity.

\begin{lemma}
\label{lemma:commute}
Suppose that $R, S$ are registers,
$Z \in R \otimes S$ is a unit vector, and $V, U_1, U_2, \ldots, U_k$ are unitary operators
on $R$ such that
\begin{enumerate}
\item Each map $U_i$ can be pushed
through $Z$ with error term $\epsilon$, and
\item The approximate
equality
$(VU_i \otimes I_S)L \underset{\epsilon}{=} (U_iV \otimes I_S)L$ holds for all $i$.
\end{enumerate}
Then,
\begin{equation}
\begin{array}{ccc}
\vcenter{\hbox{
\begin{tikzpicture}
\node[style=cstate] (z) at
(0,0) {$~~~~Z~~~~~$};
\node[style=cprocess] (uk)
at (-0.3,1) {$U_k$};
\draw (uk) to (uk|-z.north);
\node[style=cprocess] (u1)
at (-0.3,3) {$U_1$};
\draw (uk) to (uk|-z.north);
\node (dots) at (-0.3,2) {$\vdots$};
\draw (u1) to (u1|-dots.north);
\draw (dots) to (dots|-uk.north);
\node[style=cprocess] (v) at (-0.3,4) {$V$};
\draw (v) to (v|-u1.north);
\node (r) at (-0.3,5) {$R$};
\node (s) at (0.5,5) {$S$};
\draw (s) to (s|-z.north);
\draw (r) to (r|-v.north);
\end{tikzpicture}}}
& \vcenter{\hbox{{\LARGE $\underset{k^2 \epsilon}{=\joinrel=}$}}} & 
\vcenter{\hbox{
\begin{tikzpicture}
\node[style=cstate] (z) at
(0,0) {$~~~~Z~~~~~$};
\node[style=cprocess] (uk)
at (-0.3,2) {$U_k$};
\node[style=cprocess] (u1)
at (-0.3,4) {$U_1$};
\node (dots) at (-0.3,3) {$\vdots$};
\draw (u1) to (u1|-dots.north);
\draw (dots) to (dots|-uk.north);
\node[style=cprocess] (v) at (-0.3,1) {$V$};
\draw (v) to (v|-uk.south);
\node (r) at (-0.3,5) {$R$};
\node (s) at (0.5,5) {$S$};
\draw (s) to (s|-z.north);
\draw (u1) to (u1|-r.south);
\draw (v) to (v|-z.north);
\end{tikzpicture}}}
\end{array}
\end{equation}
\end{lemma}

\begin{proof}
This follows easily by $k$ applications of
Lemma~\ref{lemma:push}.
\end{proof}

By Proposition~\ref{prop:app commutativity},
for any $j \neq k$, we have
\begin{equation}
\begin{array}{ccc}
\vcenter{\hbox{
\begin{tikzpicture}
\node[style=cprocess] (l) at (1,-3.5) {$~~~~~~~L~~~~~~~$};

\node[style=bell] (bell) at (-1.2,-3.5) {};

\node[style=cprocess,minimum width = 1.5cm] (CX) at (0,-1.5) {$\mathbf{C}(X_k')$};

\draw (l.150) to[out=90,in=-90] node[left] {$A$} (CX);

\draw (bell) to[in=-90, out=45] (CX.230);
\draw (bell) to node[left] {$\overline{Q}_k$} (-1.2, 1);

\node[style=cprocess] (z) at (0,0) {$X'_\ell$};

\draw (z) to (z|-CX.north);

\draw (z) to (0,1);

\draw (l.90) to[out=90, in=-90] (2,1);
\draw (l.50) to[out=90, in=-90] (3,1);

\draw (CX.130) to[out=90,in=-90] (-0.6,1);

\end{tikzpicture}
}} 
&
\mbox{\LARGE $\underset{N\sqrt{\epsilon}}{=\joinrel=}$}
&
\vcenter{\hbox{
\begin{tikzpicture}
\node[style=cprocess] (l) at (1,-3.5) {$~~~~~~~L~~~~~~~$};

\node[style=bell] (bell) at (-1.2,-3.5) {};

\node[style=cprocess] (z) at (0,-1.5) {$X'_\ell$};

\node[style=cprocess,minimum width = 1.5cm] (CX) at (0,0) {$\mathbf{C}(X_k')$};
\draw (z) to (z|-CX.south);

\draw (l.150) to[out=90,in=-90] node[left] {$A$} (z);

\draw (bell) to[in=-90, out=75] (CX.230);
\draw (bell) to node[left] {$\overline{Q}_k$} (-1.2, 1);

\draw (CX) to (0,1);
\draw (CX.130) to[out=90,in=-90] (-0.6,1);

\draw (l.90) to[out=90, in=-90] (2,1);
\draw (l.50) to[out=90, in=-90] (3,1);

\end{tikzpicture}
}} 
\end{array}
\end{equation}
and the same holds with $(X'_\ell, X'_k)$
replaced by $(X'_\ell, Z'_k)$,
$(Z'_\ell, X'_k)$, or $(Z'_\ell, Z'_k)$.  
Also, as noted at the beginning of section~\ref{app:correct pauli}, the
gates that define $\Psi_{A,k}$ (see diagram (\ref{swapmap}) can each be pushed
through $L$ with error term $N \sqrt{\epsilon}$.  Therefore by Lemma~\ref{lemma:commute},
\begin{equation}
\begin{array}{ccc}
\vcenter{\hbox{
\begin{tikzpicture}
\node[style=cprocess] (l) at (1,-3.5) {$~~~~~~~L~~~~~~~$};

\node[style=bell] (bell) at (-1.2,-3.5) {};

\node[style=cprocess,minimum width = 1.5cm] (CX) at (0,-1.5) {$\Psi_k$};

\draw (l.150) to[out=90,in=-90] node[left] {$A$} (CX);

\draw (bell) to[in=-90, out=45] (CX.230);
\draw (bell) to node[left] {$\overline{Q}_k$} (-1.2, 1);

\node[style=cprocess] (z) at (0,0) {$X'_\ell$};

\draw (z) to (z|-CX.north);

\draw (z) to (0,1);

\draw (l.90) to[out=90, in=-90] (2,1);
\draw (l.50) to[out=90, in=-90] (3,1);

\draw (CX.130) to[out=90,in=-90] (-0.6,1);

\end{tikzpicture}
}} 
&
\mbox{\LARGE $\underset{N\sqrt{\epsilon}}{=\joinrel=}$}
&
\vcenter{\hbox{
\begin{tikzpicture}
\node[style=cprocess] (l) at (1,-3.5) {$~~~~~~~L~~~~~~~$};

\node[style=bell] (bell) at (-1.2,-3.5) {};

\node[style=cprocess] (z) at (0,-1.5) {$X'_\ell$};

\node[style=cprocess,minimum width = 1.5cm] (CX) at (0,0) {$\Psi_k$};
\draw (z) to (z|-CX.south);

\draw (l.150) to[out=90,in=-90] node[left] {$A$} (z);

\draw (bell) to[in=-90, out=75] (CX.230);
\draw (bell) to node[left] {$\overline{Q}_k$} (-1.2, 1);

\draw (CX) to (0,1);
\draw (CX.130) to[out=90,in=-90] (-0.6,1);

\draw (l.90) to[out=90, in=-90] (2,1);
\draw (l.50) to[out=90, in=-90] (3,1);

\end{tikzpicture}
}} 
\end{array}
\end{equation}
We therefore have the following, in which we first apply Proposition~\ref{prop:correct pauli} with Lemma~\ref{lemma:push}, and then apply Lemma~\ref{lemma:commute}.
\begin{equation}
\begin{array}{ccccc}
\vcenter{\hbox{
\begin{tikzpicture}

\node (out1) at (-1.5,3) {$\mathbf{Q}_{1
\ldots k}$};
\node (out2) at (-.5,3) {$A$};
\node (out3) at (.5,3) {$B$};
\node (out4) at (1.5,3) {$C$};

\node[style=cprocess] (X) at (-1.5,2) {$X_{k}$}; 

\node[style=cprocess, minimum width =1.5cm] (psi) at (-1,1) {$\Theta_k$};

\node[style=cstate,minimum width=2.5cm] (L) at (-.5,0) {$L$};

\draw (out1) to[out=-90,in=90] (out1|-X.north);
\draw (X) to (X|-psi.north);
\draw (out2) to (out2|-psi.north);
\draw (out3) to (out3|-L.north);
\draw (out4) to (out4|-L.north);

\draw (psi) to (psi|-L.north);

\end{tikzpicture}
}}
& \mbox{\LARGE $=\joinrel=$} &

\vcenter{\hbox{
\begin{tikzpicture}

\node (out1) at (-1.5,3) {$\mathbf{Q}_k$};

\node (qleft) at (-3,3) {$\mathbf{Q}_{1 \ldots k-1}$};

\node (out2) at (-.5,3) {$A$};
\node (out3) at (.5,3) {$B$};
\node (out4) at (1.5,3) {$C$};

\node[style=cprocess] (X) at (-1.5,2) {$X_{k}$}; 

\node[style=cprocess, minimum width =1.5cm] (psi) at (-1,1) {$\Psi_k$};

\node[style=cprocess] (psilower) at (-2,0) {$\Theta_{k-1}$};

\node[style=cstate,minimum width=2.5cm] (L) at (-.5,-1) {$L$};

\draw (out3) to (out3|-L.north);
\draw (out4) to (out4|-L.north);

\draw (L) to (psilower|-psilower.south);
\draw (psilower) to (psi);
\draw (psi) to (X);
\draw (psilower) to[in=-90, out=90] (qleft);

\draw (X) to (out1);

\draw (psi) to (out2);

\end{tikzpicture}
}}

\end{array}
\end{equation}

\begin{equation}
\begin{array}{ccccccc}
& \mbox{\LARGE $\underset{N^2 \sqrt{\epsilon}}{ =\joinrel=}$} &
\vcenter{\hbox{
\begin{tikzpicture}

\node (out1) at (-1.5,3) {$\mathbf{Q}_k$};

\node (qleft) at (-3,3) {$\mathbf{Q}_{1 \ldots k-1}$};

\node (out2) at (-.5,3) {$A$};
\node (out3) at (.5,3) {$B$};
\node (out4) at (1.5,3) {$C$};

\node[style=cprocess] (X) at (-1.5,1) {$X'_{k}$}; 

\node[style=cprocess, minimum width =1.5cm] (psi) at (-1,2) {$\Psi_k$};

\node[style=cprocess] (psilower) at (-2,0) {$\Theta_{k-1}$};

\node[style=cstate,minimum width=2.5cm] (L) at (-.5,-1) {$L$};

\draw (out3) to (out3|-L.north);
\draw (out4) to (out4|-L.north);

\draw (L) to (psilower|-psilower.south);
\draw (psilower) to[in=-90, out=90] (qleft);

\draw (X) to (psi);

\draw (psi) to (out2);

\draw (psi) to (out1);

\draw (psilower) to (X);

\end{tikzpicture}
}}

&  

& \mbox{\LARGE $\underset{N^3 \sqrt{\epsilon}}{ =\joinrel=}$} &
\vcenter{\hbox{
\begin{tikzpicture}

\node (out1) at (-1.5,3) {$\mathbf{Q}_k$};

\node (qleft) at (-3,3) {$\mathbf{Q}_{1 \ldots k-1}$};

\node (out2) at (-.5,3) {$A$};
\node (out3) at (.5,3) {$B$};
\node (out4) at (1.5,3) {$C$};

\node[style=cprocess] (X) at (-2,0) {$X'_{k}$}; 

\node[style=cprocess, minimum width =1.5cm] (psi) at (-1,2) {$\Psi_k$};

\node[style=cprocess] (psilower) at (-1.8,1) {$\Theta_{k-1}$};

\node[style=cstate,minimum width=2.5cm] (L) at (-.5,-1) {$L$};

\draw (out3) to (out3|-L.north);
\draw (out4) to (out4|-L.north);

\draw (L) to (X.290);
\draw (psilower.120) to[in=-90, out=90] (qleft);

\draw (X) to (psilower);

\draw (psi) to (out2);

\draw (psi) to (out1);

\draw (psilower) to (psi);

\end{tikzpicture}
}}
\end{array}
\end{equation}
An analogous statement holds with $X$ replaced by $Z$.
Applying the isometries $\Psi_{k+1}, \ldots, \Psi_N$ in order now completes the proof.

\end{appendices}

\end{document}

%% file: tikzstyles.tex
\makeatletter
\usetikzlibrary{circuits.ee.IEC}

\def\declaretrapezoid#1#2#3#4#5{
\pgfdeclareshape{#1}{
\inheritsavedanchors[from=rectangle]
\inheritanchorborder[from=rectangle]
\inheritanchor[from=rectangle]{center}
\inheritanchor[from=rectangle]{north}
\inheritanchor[from=rectangle]{south}
\inheritanchor[from=rectangle]{east}
\inheritanchor[from=rectangle]{west}
\inheritanchor[from=rectangle]{south east}
\inheritanchor[from=rectangle]{north west}
\inheritanchor[from=rectangle]{north east}
\inheritanchor[from=rectangle]{south west}

\backgroundpath{
  \pgfextractx{\pgf@xa}{\southwest}
  \pgfextracty{\pgf@ya}{\southwest}
  \pgfextractx{\pgf@xb}{\northeast}
  \pgfextracty{\pgf@yb}{\northeast}

  \pgfpathmoveto{\pgfpoint{\pgf@xa-#2*(\pgf@yb-\pgf@ya)}{\pgf@ya}}
  \pgfpathlineto{\pgfpoint{\pgf@xa-#3*(\pgf@yb-\pgf@ya)}{\pgf@yb}}
  \pgfpathlineto{\pgfpoint{\pgf@xb+#4*(\pgf@yb-\pgf@ya)}{\pgf@yb}}
  \pgfpathlineto{\pgfpoint{\pgf@xb+#5*(\pgf@yb-\pgf@ya)}{\pgf@ya}}
  \pgfpathlineto{\pgfpoint{\pgf@xa-#2*(\pgf@yb-\pgf@ya)}{\pgf@ya}}
  }
  }
}

\pgfdeclareshape{clipped triangle}{
\inheritsavedanchors[from=rectangle]
\inheritanchorborder[from=rectangle]
\inheritanchor[from=rectangle]{center}
\inheritanchor[from=rectangle]{north}
\inheritanchor[from=rectangle]{south}
\inheritanchor[from=rectangle]{east}
\inheritanchor[from=rectangle]{west}
\inheritanchor[from=rectangle]{south east}
\inheritanchor[from=rectangle]{north west}
\inheritanchor[from=rectangle]{north east}
\inheritanchor[from=rectangle]{south west}

\backgroundpath{
  \pgfextractx{\pgf@xa}{\southwest}
  \pgfextracty{\pgf@ya}{\southwest}
  \pgfextractx{\pgf@xb}{\northeast}
  \pgfextracty{\pgf@yb}{\northeast}

  \pgfpathmoveto{\pgfpoint{\pgf@xa}{\pgf@ya}}
  \pgfpathlineto{\pgfpoint{\pgf@xb+0.5*(\pgf@xb - \pgf@xa}{\pgf@ya}}
  \pgfpathlineto{\pgfpoint{0.6*\pgf@xa + 0.4*\pgf@xb)}{\pgf@yb+0.5*(\pgf@yb-\pgf@ya)}}
  \pgfpathlineto{\pgfpoint{\pgf@xa}{\pgf@yb}}
  \pgfpathlineto{\pgfpoint{\pgf@xa}{\pgf@ya}}
  }
}

\pgfdeclareshape{clipped triangle xflip}{
\inheritsavedanchors[from=rectangle]
\inheritanchorborder[from=rectangle]
\inheritanchor[from=rectangle]{center}
\inheritanchor[from=rectangle]{north}
\inheritanchor[from=rectangle]{south}
\inheritanchor[from=rectangle]{east}
\inheritanchor[from=rectangle]{west}
\inheritanchor[from=rectangle]{south east}
\inheritanchor[from=rectangle]{north west}
\inheritanchor[from=rectangle]{north east}
\inheritanchor[from=rectangle]{south west}

\backgroundpath{
  \pgfextractx{\pgf@xa}{\southwest}
  \pgfextracty{\pgf@ya}{\southwest}
  \pgfextractx{\pgf@xb}{\northeast}
  \pgfextracty{\pgf@yb}{\northeast}

  \pgfpathmoveto{\pgfpoint{\pgf@xb}{\pgf@ya}}
  \pgfpathlineto{\pgfpoint{\pgf@xa+0.5*(\pgf@xa - \pgf@xb}{\pgf@ya}}
  \pgfpathlineto{\pgfpoint{0.6*\pgf@xb + 0.4*\pgf@xa)}{\pgf@yb+0.5*(\pgf@yb-\pgf@ya)}}
  \pgfpathlineto{\pgfpoint{\pgf@xb}{\pgf@yb}}
  \pgfpathlineto{\pgfpoint{\pgf@xb}{\pgf@ya}}
  }
}

\pgfdeclareshape{clipped triangle xyflip}{
\inheritsavedanchors[from=rectangle]
\inheritanchorborder[from=rectangle]
\inheritanchor[from=rectangle]{center}
\inheritanchor[from=rectangle]{north}
\inheritanchor[from=rectangle]{south}
\inheritanchor[from=rectangle]{east}
\inheritanchor[from=rectangle]{west}
\inheritanchor[from=rectangle]{south east}
\inheritanchor[from=rectangle]{north west}
\inheritanchor[from=rectangle]{north east}
\inheritanchor[from=rectangle]{south west}

\backgroundpath{
  \pgfextractx{\pgf@xa}{\southwest}
  \pgfextracty{\pgf@ya}{\southwest}
  \pgfextractx{\pgf@xb}{\northeast}
  \pgfextracty{\pgf@yb}{\northeast}

  \pgfpathmoveto{\pgfpoint{\pgf@xb}{\pgf@yb}}
  \pgfpathlineto{\pgfpoint{\pgf@xa+0.5*(\pgf@xa - \pgf@xb}{\pgf@yb}}
  \pgfpathlineto{\pgfpoint{0.6*\pgf@xb + 0.4*\pgf@xa)}{\pgf@ya+0.5*(\pgf@ya-\pgf@yb)}}
  \pgfpathlineto{\pgfpoint{\pgf@xb}{\pgf@ya}}
  \pgfpathlineto{\pgfpoint{\pgf@xb}{\pgf@yb}}
  }
}

\pgfdeclareshape{clipped triangle yflip}{
\inheritsavedanchors[from=rectangle]
\inheritanchorborder[from=rectangle]
\inheritanchor[from=rectangle]{center}
\inheritanchor[from=rectangle]{north}
\inheritanchor[from=rectangle]{south}
\inheritanchor[from=rectangle]{east}
\inheritanchor[from=rectangle]{west}
\inheritanchor[from=rectangle]{south east}
\inheritanchor[from=rectangle]{north west}
\inheritanchor[from=rectangle]{north east}
\inheritanchor[from=rectangle]{south west}

\backgroundpath{
  \pgfextractx{\pgf@xa}{\southwest}
  \pgfextracty{\pgf@ya}{\southwest}
  \pgfextractx{\pgf@xb}{\northeast}
  \pgfextracty{\pgf@yb}{\northeast}

  \pgfpathmoveto{\pgfpoint{\pgf@xa}{\pgf@yb}}
  \pgfpathlineto{\pgfpoint{\pgf@xb+0.5*(\pgf@xb - \pgf@xa}{\pgf@yb}}
  \pgfpathlineto{\pgfpoint{0.6*\pgf@xa + 0.4*\pgf@xb)}{\pgf@ya+0.5*(\pgf@ya-\pgf@yb)}}
  \pgfpathlineto{\pgfpoint{\pgf@xa}{\pgf@ya}}
  \pgfpathlineto{\pgfpoint{\pgf@xa}{\pgf@yb}}
  }
}

\declaretrapezoid{right-in trapezoid}{0}{0}{0.5}{0}
\declaretrapezoid{right-out trapezoid}{0}{0}{0}{0.5}
\declaretrapezoid{left-in trapezoid}{0}{0.5}{0}{0}
\declaretrapezoid{left-out trapezoid}{0.5}{0}{0}{0}

\tikzstyle{terminal}=[circuit ee IEC, ground, thick, rotate=-90]

\tikzstyle{cedge}=[swap]
\tikzstyle{qedge}=[swap,very thick]

\tikzstyle{qprocess}=[draw,very thick,shape=right-in trapezoid]
\tikzstyle{qprocessxflip}=[draw,very thick,shape=left-in trapezoid]
\tikzstyle{qprocessyflip}=[draw,very thick,shape=right-out trapezoid]
\tikzstyle{qprocessxyflip}=[draw,very thick,shape=left-out trapezoid]

\tikzstyle{qstate}=[draw,very thick,shape=clipped triangle yflip]
\tikzstyle{qstatexflip}=[draw, very thick,shape=clipped triangle xyflip]
\tikzstyle{qeffectxflip}=[draw, very thick,shape=clipped triangle xflip]
\tikzstyle{qeffect}=[draw, very thick,shape=clipped triangle]

\tikzstyle{cstate}=[draw,shape=clipped triangle yflip]
\tikzstyle{cstatexflip}=[draw, shape=clipped triangle xflip]
\tikzstyle{ceffectxflip}=[draw, shape=clipped triangle xyflip]
\tikzstyle{ceffect}=[draw, shape=clipped triangle]

\tikzstyle{cprocess}=[draw,shape=right-in trapezoid]
\tikzstyle{cprocessxflip}=[draw,shape=left-out trapezoid]
\tikzstyle{cprocessyflip}=[draw,shape=right-out trapezoid]
\tikzstyle{cprocessxyflip}=[draw,shape=left-in trapezoid]

\tikzstyle{spider}=[draw,shape=circle]

\tikzstyle{bell}=[draw,shape=circle,fill=gray]

\tikzstyle{hadamard}=[draw,very thick,shape=rectangle]